\documentclass[final,onefignum,onetabnum]{siamonline220329}

\usepackage{lipsum}
\usepackage{amsfonts}
\usepackage{graphicx}
\usepackage{epstopdf}
\usepackage{algorithmic}
\usepackage{amsmath,amssymb}
\usepackage{bbm}
\ifpdf
  \DeclareGraphicsExtensions{.eps,.pdf,.png,.jpg}
\else
  \DeclareGraphicsExtensions{.eps}
\fi

\usepackage{enumitem}
\setlist[enumerate]{leftmargin=.5in}
\setlist[itemize]{leftmargin=.5in}

\newsiamremark{remark}{Remark}
\newsiamremark{hypothesis}{Assumption}
\crefname{hypothesis}{Assumption}{Assumption}
\newsiamthm{claim}{Claim}

\allowdisplaybreaks

\headers{Posterior Sampling with Spike-and-Slab Prior}{Q. Jiang}

\title{From Estimation to Sampling for Bayesian Linear Regression with Spike-and-Slab Prior\thanks{Submitted to the editors DATE.
\funding{This work is supported under DOE.}}}

\author{Qijia Jiang\thanks{Lawrence Berkeley National Laboratory
  (\email{qjiang@lbl.gov}).}}

\usepackage{amsopn}

\begin{document}


\maketitle

\begin{abstract}
We consider Bayesian linear regression with sparsity-inducing prior and design efficient sampling algorithms leveraging posterior contraction properties. A quasi-likelihood with Gaussian spike-and-slab (that is favorable both statistically and computationally) is investigated and two algorithms based on Gibbs sampling and Stochastic Localization are analyzed, both under the same (quite natural) statistical assumptions that also enable valid inference on the sparse planted signal. The benefit of the Stochastic Localization sampler is particularly prominent for data matrix that is not well-designed.  
\end{abstract}



\section{Introduction}

In this work we study posterior sampling arising from high-dimensional Bayesian variable selection -- our focus is on sampling from the full posterior for uncertainty quantification purpose as opposed to computing aspect of it (e.g., point estimators). Given design matrix $X\in \mathbb{R}^{n\times p}$ and response $y\in \mathbb{R}^n$, the linear regression model with Spike-and-Slab prior has posterior 
\begin{equation}
\label{eqn:general}
\pi(\beta|y,X)\propto \mathcal{L}(y|X,\beta)\mathbb{P}_{\text{prior}}(\beta)\propto \exp(-\frac{1}{2\sigma^2}\|y-X\beta\|_2^2) (\otimes_{i=1}^p(1-z)G_0(\beta_i)+zG_1(\beta_i))
\end{equation}
for some $z\in (0,1)$, where $G_0$ has density more concentrated around $0$ than $G_1$. What makes the Bayesian methodology attractive is that it comes with credible sets instead of a single summary statistics; however, we emphasize that we will study Bayesian guarantee in a frequentist framework in this paper, where we assume there is a planted (and \emph{fixed}) $k$-sparse signal $\beta^*$ for which data is generated from, i.e., $y=X\beta^*+\epsilon$ for $\epsilon\sim \mathcal{N}(0,\sigma^2 I)$. This prior can be viewed as a regularized least squares / penalized likelihood if one draws parallel to the frequentist perspective, where Lasso ($\ell_1$ penalty) corresponds to the posterior mode of i.i.d Laplace($\lambda$) prior with density $\frac{\lambda}{2}\exp(-\lambda|\beta|)$:
\[\hat{\beta}_{\text{Lasso}}\leftarrow\arg\min_\beta \; \|y-X\beta\|_2^2+\lambda\sum_{i=1}^p |\beta_i|\, .\]
Lasso, however, isn't fully Bayesian in the sense credible interval building upon the posterior distribution does not provide valid coverage guarantee \cite[Theorem 7]{castillo2015regression} for $\beta^*$. Therefore good performance of posterior mode doesn't automatically translate to good performance of the full posterior. This is, in some sense, not surprising since it has to balance between the task of selection and prediction (i.e., shrinkage and bias). Spike and Slab prior, on the other hand, by explicitly introducing two scales/groups, is better at dealing with this tension. Indeed, favorable statistical properties can be established on the posterior for inference on the unknown sparse $\beta^*$ -- in what follows, we will design sampling procedures under statistical assumptions for the model and will be mostly concerned with the scaling with $p$ when it comes to computational methods. We note that for the purpose of recovering the sparse $\beta^*$, classical BvM says that data will eventually wash out the influence of the prior choice, however mismatch between the prior and the truth will be reflected in the slow posterior contraction rate of $\pi_n(\beta |y^n)\rightarrow \delta_{\beta^*}$ in terms of statistical efficiency. Another way to see this manifested is through the variational inequality 
\[-\log\mathbb{E}_{\text{prior} (\beta)} \mathcal{L}_{y,X}(\beta)=\min_{\rho\ll \mathbb{P}_{\text{prior}}(\beta)} \{-\mathbb{E}_\rho [\log\mathcal{L}_{y,X}(\beta)]+\text{KL}(\rho || \mathbb{P}_{\text{prior}} (\beta))\}\]
and the minimizer $\rho^*$ is precisely \eqref{eqn:general} when $\mathcal{L}_{y,X}(\beta)$ is the likelihood function, therefore the posterior will concentrate on maximizers of the likelihood in presence of the evidence from data, while staying faithful to the prior knowledge one may have.

\subsection{Related Literature}
Statistical properties of \eqref{eqn:general} have been studied by \cite{castillo2015regression,narisetty2014bayesian,spike-and-slab-lasso} with different choices for $G_0,G_1,z,\sigma$. On a closely related prior, computational-statistical guarantees given by \cite{yang2016computational} highlight that sharp concentration of the high-dimensional posterior distribution (i.e., $\pi_n(z^*|y)\gtrsim 1-p^{-1}$ with probability at least $1-p^{-c}$ assuming smallest non-zero element of $\beta^* \gtrsim\sigma^2\log p/n$) need not lead to polynomial mixing of MCMC algorithm. \emph{Unless} one restricts the size of the state space the prior is supported on $\mathbbm{1}\{\|z\|_0\leq u\}$, the authors show that the gradient-free Metropolis-Hastings algorithm (also known as Add-Delete-Swap in this context) can have mixing time scaling exponentially with $p$. However, this upper bound $u$ depends on quantities unknown in practice. Gibbs sampler is widely used for spike-and-slab models, and its convergence is analyzed in \cite{atchade2021approximate} with numerical speedup investigated in \cite{biswas2022scalable}. Various approximate schemes exist, where in \cite{ray2022variational} mean-field variational inference ideas are used (i.e., reduce model search space from $2^p$ to $p$ assuming coordinates are independent) to show posterior contraction but since the objective to be optimized is non-convex, guarantee for convergence to global optima is hard to establish (in fact it was empirically observed that the result can be sensitive to initialization). Some previous attempts also focus on designing efficient algorithms for computing point estimators such as posterior modes using e.g., EM algorithms for priors with continuous support \cite{rovckova2014emvs}. 




The philosophy we adopt for sampling from the non-log-concave spike-and-slab posterior \eqref{eqn:general} is close in spirit to (1)\cite{belloni2009computational}, where posterior converges to a normal limit as both the sample size $n$ and parameter dimension $p$ grow to infinity at appropriate rate (reminiscent of Bernstein-von-Mises theorem which states the posterior approach a Gaussian centered at MLE with Fisher information covariance under appropriate assumptions), and show polynomial time mixing in $p$ -- an assumption on the starting point for the algorithm that falls in the approximate support of the posterior, i.e., where CLT applies, is also imposed; (2) A line of investigation on Bayesian nonlinear inverse problem \cite{nickl} also crucially hinges on warm start into the locally convex region where most of the posterior mass concentrates for polynomial-time convergence of the MCMC algorithm they design. On the other hand, standard off-the-shelf gradient-based HMC, MALA samplers typically struggle for potentials deviating significantly from log-concavity beyond functional inequalities 
 -- one could check that the Log-Sobolev constant (therefore mixing time) scales exponentially with the separation between the peaks, in addition to already expensive gradient calculation, without the possible help of parallel tempering/replica exchange that avoids being trapped in separated modes. In fact, these are not surprising in light of the asymptotic posterior shape characterization in \cite[Theorem 6]{castillo2015regression} where they are shown to be well-approximated by random (i.e., data-dependent) mixture of Gaussians.
 
 \subsection{Notation \& Outline}
(In)equalities with $\lesssim, \gtrsim, \asymp$ hold up to absolute constants. For two models $z,z' \in \{0,1\}^p$, $z\subset z'$ means that the active components of $z$ is a subset of that of $z'$, and $\|z\|_0$ counts the number of non-zeros/active elements. We write $j\notin z$ to indicate $z_j=0$. Total-variation distance is defined as $\|\mu-\nu\|_{\text{TV}}=\sup_{A\in \mathcal{B}} |\mu(A)-\nu(A)|\in [0,1]$, and Wasserstein-2 distance is defined as $W_2(\mu,\nu) = \inf_{x\sim \mu, y\sim \nu} \mathbb{E}[\|x-y\|^2]^{1/2}$, which satisfies triangle inequality. Moreover, we use $o_n(1)$ to specify a quantity tending to $0$ as $n \rightarrow \infty$, and $O_p(a)$ for the usual stochastic boundedness. Both $X_n\xrightarrow{P} X$ and $\text{p-}\lim_{n\rightarrow \infty} X_n = X$ denote convergence in probability. In what follows, \Cref{sec:meta_gibbs} studies Gibbs sampler, \Cref{sec:SL-meta} the Stochastic Localization Sampler, both under warm start and posterior contraction assumptions. These statistical assumptions are justified in \Cref{sec:statistics} for the particular quasi-likelihood posterior with continuous spike-and-slab prior that we focus on in this work.

\section{(Scalable) Gibbs Sampler}
\label{sec:meta_gibbs}
In this section, we (1) give Gibbs update and efficient implementation for point-mass-like spike-and-slab priors, along with its random design analogue for Gaussian design matrix; (2) provide mixing guarantee from a warm start. We also highlight the bottleneck for Gibbs-based samplers for this class of posteriors.

\subsection{Point-mass-like Spike-and-Slab}
A popular approach of conducting Bayesian variable selection in the regime $p \gg n$ is through setting up a hierarchical model: for linear model $y=X\beta+\epsilon$ with $\epsilon\sim\mathcal{N}(0,\sigma^2 I_n)$ for $\sigma^2$ the noise variance (inverse Gamma distribution on $\sigma^2$ is sometimes considered but we will assume that it's known here) and the sparsity prior $z_j\sim\text{Bern}(q)$ where $\beta_j\vert z_j \sim z_j\mathcal{N}(0,\tau_1^2)+(1-z_j)\delta_0(\beta_j)$ for all $j\in[p]$, the joint posterior is 
\[\pi(\beta,z\vert y)\propto \mathcal{N}(y;X\beta,\sigma^2 I_n)\prod_{j=1}^p (q\mathcal{N}(\beta_j;0,\tau_1^2))^{z_j}\cdot((1-q)\delta_0(\beta_j))^{1-z_j}\, .\]
The Gibbs update, which relies on the availability of conditional probabilities, becomes
\begin{align*}
\pi (\beta\vert z,y) &\propto \mathcal{N}(y;X\beta,\sigma^2)\prod_{j=1}^p (\delta_0(\beta_j))^{1-z_j}\cdot(\mathcal{N}(\beta_j;0,\tau_1^2))^{z_j}\\
&\propto \exp\left(-\frac{1}{2\sigma^2}(\beta^\top X^{\top}X\beta-2\beta^\top X^\top y)-\beta^\top D(\frac{z_j}{2\tau_1^2})\beta\right)\prod_{j=1}^p (\delta_0(\beta_j))^{1-z_j}\\
&\sim \mathcal{N}(\bar{\beta};\Sigma^{-1}\bar{X}^\top y,\sigma^2\Sigma^{-1}) \prod_{j=1}^p (\delta_0(\beta_j))^{1-z_j}
\end{align*}
for $\Sigma(z) = \bar{X}^\top \bar{X}+2\sigma^2 D(\frac{z_j}{2\tau_1^2})$, where $\bar{X}$ denotes the $n\times \|z\|_0$ sub-matrix with $z_j=1$, $\bar{\beta}$ the subvector with active coordinates, and $D(\cdot)$ a $\|z\|_0\times \|z\|_0$ diagonal matrix with the indicated components. In other words, 
\begin{equation}
\label{eqn:point_mass_beta}
\pi (\beta\vert z,y)\sim \mathcal{N}\left(\bar{\beta}; (\bar{X}^\top \bar{X}+\frac{\sigma^2}{\tau_1^2}I)^{-1}\bar{X}^\top y, \sigma^2(\bar{X}^\top \bar{X}+\frac{\sigma^2}{\tau_1^2}I)^{-1}\right)\otimes \prod_{j=1}^p (\delta_0(\beta_j))^{1-z_j}
\end{equation}
where $\delta_0(\beta_j)$ denotes Dirac delta, i.e., $\beta_j=0$ if $z_j=0$. The conditional distribution for $z$ is
\begin{align}
\pi(z\vert \beta,y) &\propto \prod_{j=1}^p (q\mathcal{N}(\beta_j;0,\tau_1^2))^{z_j}\cdot((1-q)\delta_0(\beta_j))^{1-z_j} \nonumber \\
&\sim \prod_{j=1}^p \text{Bern}\left(z_j;\frac{q\mathcal{N}(\beta_j;0,\tau_1^2)}{(1-q)\delta_0(\beta_j)+q\mathcal{N}(\beta_j;0,\tau_1^2)}\right) \label{eqn:point_mass_z}
\end{align}
which suggests $z_j=0$ if $\beta_j=0$ and $z_j=1$ if $\beta_j\neq 0$. It might be tempting to conclude that this is computationally favorable as \eqref{eqn:point_mass_beta} involves inversion of a lower-dimensional matrix, as opposed to a continuous prior of the form $\beta_j\vert z_j \sim z_j\mathcal{N}(0,\tau_1^2)+(1-z_j)\mathcal{N}(0,\tau_0^2)$, that necessarily requires matrix inversion of size $p\times p$. However, the updates \eqref{eqn:point_mass_beta}-\eqref{eqn:point_mass_z} in fact lead to a non-convergent / reducible Markov chain, i.e., the chain gets stuck whenever it generates $\beta_j=0$, although statistically the posterior on $\beta$ contracts at the \emph{near minimax-optimal} rate for the recovery of $\beta^*$. For example for a related prior where $\beta_1,\dots,\beta_p$ are i.i.d from $(1-r)\delta_0+r\text{Laplace}$, and $r\sim \text{Beta}(1,p^u)$ hyper prior with $u>1$,
the important work of \cite{castillo2015regression} showed under $k_n$-sparse compatibility assumption on the design matrix for the high-dimensional setting $p>n$, uniformly over $k_n$-sparse signals, 
\begin{equation}
\label{eqn:gold-standard}
\sup_{\|\beta^*\|_0\leq k_n}\mathbb{E}_{\beta^*}\left[\pi_n(\beta\colon \|\beta-\beta^*\|_1 \gtrsim k_n\sqrt{\log p}/\|X\|\,\vert\, y^n)\right]\xrightarrow{n\rightarrow \infty} 0\,.
\end{equation}
Note this is a remarkably strong statement about the \emph{complete} posterior $\pi(\cdot|y)$, which is a random measure over $\beta$ for any fixed $\beta^*$, and not just aspect of it such as the posterior mode / mean as   
\[\sup_{\|\beta^*\|_0\leq k_n} \mathbb{E}_{\beta^*}\left[ \left\|\int \beta \pi(\beta\,|\, y^n) \, d\beta -\beta^*\right\|^2 \right]\lesssim 2k_n\log(p/k_n)\,, \]
which the Lasso estimator $\hat{\beta}_{\text{Lasso}}$ also verify with an appropriate choice of $\lambda$. Above $k_n\rightarrow \infty$ is permitted as $n\rightarrow \infty$.

For this reason, computational strategies involving \emph{exact-sparsity} inducing priors resort to Add-Delete-Swap or shotgun stochastic search \cite{shotgun}, which integrate out the regression coefficients from the posterior (i.e., design samplers based on $\mathbb{P}(z|y)$ over $\{0,1\}^p$), but falls short of solving both the variable selection ($z$) and parameter estimation ($\beta$) problems simultaneously. On the other hand, Gibbs can handle spike-and-slab prior with continuous support effortlessly, that doesn't have this trans-dimensionality problem, but inversion of a $p\times p$ matrix renders the sampling procedure expensive. The quasi-likelihood approach below, which is a variant of the classical formulation \eqref{eqn:general}, provides a middle ground that balance between the desirable statistical performance and computational convenience, as we will elaborate. 


\begin{proposition}
\label{prop:sparsify}
The sparsified likelihood \cite{atchade2018approach} that has posterior (with $\tau_1 \gg \tau_0$)
\begin{equation}
\label{eqn:quasi_posterior}
\pi(\beta,z\vert y)\propto \mathcal{N}(y;X_{z}\beta_{z},\sigma^2 I_n)\prod_{j=1}^p (q\mathcal{N}(\beta_j;0,\tau_1^2))^{z_j}\cdot((1-q)\mathcal{N}(\beta_j;0,\tau_0^2))^{1-z_j}
\end{equation}
targets a different posterior than Skinny Gibbs \cite{narisetty2018skinny}, but can also be sampled using Gibbs with a reduced-dimensional matrix inversion operation at each iteration.  
\end{proposition}
\begin{proof}
For the posterior with quasi-likelihood, we alternate between
\begin{align}
\pi (\beta\vert z,y) 
&\propto \exp\left(-\frac{1}{2\sigma^2}(\bar{\beta}^\top \bar{X}^{\top}\bar{X}\bar{\beta}-2\bar{\beta}^\top \bar{X}^\top y)-\bar{\beta}^\top D(\frac{1}{2\tau_1^2})\bar{\beta}\right)\prod_{j=1}^p (\mathcal{N}(\beta_j;0,\tau_0^2))^{1-z_j}\nonumber \\
&\sim \mathcal{N}(\beta_z;\Sigma^{-1}\bar{X}^\top y,\sigma^2\Sigma^{-1}) \prod_{j=1}^p (\mathcal{N}(\beta_j;0,\tau_0^2))^{1-z_j} \label{eqn: beta_update}
\end{align}
where $\Sigma(z) = \bar{X}^\top \bar{X}+2\sigma^2 D(\frac{z_j}{2\tau_1^2})$ and for each $j\in[p]$ sequentially
\begin{align}
\pi(z_j\vert \beta,y,z_{-j}) &\propto \prod_{j=1}^p (q\mathcal{N}(\beta_j;0,\tau_1^2))^{z_j}\cdot((1-q)\mathcal{N}(\beta_j;0,\tau_0^2))^{1-z_j}\mathcal{N}(y;X_{z}\beta_{z},\sigma^2) \nonumber\\
&\propto  ((1-q)\mathcal{N}(\beta_j;0,\tau_0^2))^{1-z_j}\cdot q^{z_j}\times \mathcal{N}(\beta_{z};\Sigma^{-1}\bar{X}^\top y,\sigma^2\Sigma^{-1}) \nonumber\\
&\sim \text{Bern}\left(z_j;\frac{q\mathcal{N}(\beta_{z};\Sigma^{-1}\bar{X}^\top y,\sigma^2\Sigma^{-1})}{(1-q)\mathcal{N}(\beta_j;0,\tau_0^2)+q\mathcal{N}(\beta_{z};\Sigma^{-1}\bar{X}^\top y,\sigma^2\Sigma^{-1})}\right) \label{eqn:z_update}
\end{align}
which although is still Bernoulli, is no longer independent across coordinates, and the update for $z_j$ depends on not just $\beta_j$. In \eqref{eqn:z_update}, the normal distribution in the numerator involves setting $z_j=1$ and the rest as the conditioned $z_{\backslash j}$ at the current iteration. Another way to write the update for $z_j$ conditional on the rest is
\begin{align*}
&Q_j:=\frac{\pi(z_j=1\vert \beta,y,z_{-j})}{\pi(z_j=0\vert \beta,y,z_{-j})}=\frac{q}{1-q} \frac{\tau_0}{\tau_1}\times\\
&\frac{\exp[-(\beta_{z}-(\bar{X}^\top \bar{X}+\frac{\sigma^2}{\tau_1^2}I)^{-1}\bar{X}^\top y)^{\top}\frac{1}{2\sigma^2}(\bar{X}^\top \bar{X}+\frac{\sigma^2}{\tau_1^2}I)(\beta_{z}-(\bar{X}^\top \bar{X}+\frac{\sigma^2}{\tau_1^2}I)^{-1}\bar{X}^\top y)]}{\exp(-\beta_j^2/2\tau_0^2)}\\
&\propto \frac{q}{1-q}\frac{\tau_0}{\tau_1}\frac{\exp[-\frac{1}{2\sigma^2}\beta_{z}^{\top}(\bar{X}^\top\bar{X}+\frac{\sigma^2}{\tau_1^2}I)\beta_{z}+\frac{1}{\sigma^2}y^\top \bar{X}\beta_{z}]}{\exp(-\beta_j^2/2\tau_0^2)}\\
&\propto \frac{q}{1-q}\frac{\tau_0}{\tau_1}\frac{\exp(-\beta_j^2(1/2\tau_1^2+(X^{\top} X)_{jj}/2\sigma^2))}{\exp(-\beta_j^2/2\tau_0^2)}\exp(-\frac{1}{\sigma^2}\beta_j X_j^\top \bar{X}_{\backslash j}\beta_{z,\backslash j}+\frac{1}{\sigma^2}\beta_j X_j^\top y)\\
&=: \Pi_j \cdot \exp(-\beta_j^2(X^\top X)_{jj}/2\sigma^2)
\end{align*}
where $\bar{X}_{\backslash j}$ denotes the submatrix corresponding to the components of $z_{\backslash j}$ such that $z_k=1$. Note that $Q_j$ doesn't depend on $z_j$. 

This is slightly different from Skinny Gibbs update, which approximate the covariance matrix (ignoring cross-correlation between active $\bar{X}$ and inactive $\bar{X}_c$ components)
\[\begin{bmatrix}
\bar{X}^\top \bar{X}+\frac{\sigma^2}{\tau_1^2}I  &\bar{X}^\top \bar{X}_c\\
\bar{X}^{\top}_c\bar{X} & \bar{X}^{\top}_c \bar{X}_c+\frac{\sigma^2}{\tau_0^2}I 
\end{bmatrix} \quad \text{with}\quad 
\begin{bmatrix}
\bar{X}^\top \bar{X}+\frac{\sigma^2}{\tau_1^2}I  & 0\\
0 & \text{Diag}(\bar{X}^{\top}_c \bar{X}_c)+\frac{\sigma^2}{\tau_0^2}I 
\end{bmatrix} \]
therefore the update for $\beta_j$ for which $z_j=0$, although independent across coordinates, would involve $\text{Diag}(\bar{X}^{\top}_c \bar{X}_c)$ for the inactive components (but the update for the active components are the same as \eqref{eqn: beta_update}), and the update for $z_j$ in this case can be shown to be $\Pi_j$ (c.f. (3.12) in \cite{handbook2021}). However Skinny Gibbs posterior \cite{narisetty2018skinny} still enjoys strong model selection consistency property $\pi(z=z^*|y )\rightarrow 1$ asymptotically, as $p_n>n$ both grow at a proportional ratio.  
\end{proof}
\begin{remark}
One might also consider a Hogwild asynchronous style update with all $z_j$ drawn in parallel, using the latest $z$ in the shared memory with possible overwriting, although it seems hard to characterize the error introduced by this approximate MCMC scheme. If all the updates use the $z$ from the previous iteration, it amounts to assuming that the $z_j$'s are independent. 
\end{remark}

We'd like to mention that a Metropolized-Gibbs strategy with an accept/reject implementation for the $\{z_j\}_{j=1}^p$ update on \eqref{eqn:quasi_posterior} was proposed in \cite{atchade2018approach}, but we find the algorithm above somewhat more natural.

\begin{remark}
Such Gibbs update based on sparsified-likelihood can also be generalized to spike/slab distributions that admit representation as a scale-mixture of normals: for example in the case when $G_0/G_1$ is Laplace, one could write for $\lambda>0$
\[\frac{\lambda}{2}e^{-\lambda|\beta|}=\int_0^\infty \frac{1}{\sqrt{2\pi s}}e^{-\frac{\beta^2}{2s}}\frac{\lambda^2}{2}e^{-\lambda^2 s/2}\, ds,\]
which is equivalent to having $\beta|s \sim \mathcal{N}(0,
\sqrt{s}), s\sim \text{Laplace}(\lambda^2)$, and one can alternate between updating $\beta,z,s$; the conditional distribution of $s$ will be an inverse-Gamma in this case.
\end{remark}

\subsubsection{Practical Matters}
\label{sec:practical}
The update given in \cref{prop:sparsify} requires drawing samples from a multivariate Gaussian with covariance matrix that involves inversion of a $\|z\|_0 \times \|z\|_0$ matrix (since the posterior is concentrated on sparse $z$'s as we will show in \Cref{subsec:posterior_contract}, one can expect $\|z\|_0 \ll p$). This is the more expensive step among \eqref{eqn: beta_update}-\eqref{eqn:z_update}. Building on the work of \cite{biswas2022scalable,fast_conditional_gaussian}, data augmentation and pre-computation can be used to improve the \eqref{eqn: beta_update} step as follows, which cost $\mathcal{O}(\max\{n^2\|z\|_0,n^3\})$ since forming the matrix takes $\mathcal{O}(n^2\|z\|_0)$ and inverting takes $\mathcal{O}(n^3)$.
\begin{algorithm}[H]
\caption{Sample from $\mathcal{N}(\bar{\beta};\Sigma^{-1}\bar{X}^\top y,\sigma^2\Sigma^{-1})$ for $\Sigma(z) = \bar{X}^\top \bar{X}+\sigma^2/\tau_1^2 \cdot D_t(z)$ where $D_t$ is $\|z\|_0\times \|z\|_0$ diagonal and $\bar{X}$ is $n\times \|z\|_0$ consisting of active variables}\label{alg:sample_gaussian}
\begin{algorithmic}
\STATE{Sample $r\sim\mathcal{N}(0,D_t^{-1}), \zeta\sim \mathcal{N}(0,I_n)$}
\STATE{Set $v=\bar{X}r+\zeta$}
\STATE{Compute $u=(I_n+\bar{X}D_t^{-1}\bar{X}^\top)^{-1}(\frac{1}{\sigma}y-v)=: M_t^{-1}(\frac{1}{\sigma}y-v)$}
\RETURN $\bar{\beta}=\sigma(r+D_t^{-1}\bar{X}^\top u)$
\end{algorithmic}
\end{algorithm}
If the number of variables switching states between consecutive iterations is small (i.e., $\|z_t-z_{t-1}\|_0$ small, either due to sparse $z$/posterior concentration from \cref{prop:stat} or stable Markov chain), a few more ideas can be used for speeding up \cref{alg:sample_gaussian}:
\begin{enumerate}
\item Use the previous $M_t\in \mathbb{R}^{n\times n}$ as preconditioner and solve the linear system using conjugate gradient, which only involves matrix-vector product
\item Instead of computing $M_t^{-1}$ from scratch at every step, perform Sherman-Morrison on the previous matrix $M_{t-1}$, since only a few columns are added/deleted
\end{enumerate}

Per-iteration cost aside, due to the curse of dimensionality, blocked updates can also help with mixing as illustrated by the following example. From \cref{prop:sparsify} we know
\begin{align*}
&\frac{\pi(z_j=1\vert \beta,y,z_{-j})}{\pi(z_j=0\vert \beta,y,z_{-j})}\propto \\
&\frac{q}{1-q}\frac{\tau_0}{\tau_1}\exp(-\frac{1}{2}(1/\tau_1^2-1/\tau_0^2)\beta_j^2)\exp(-\frac{1}{\sigma^2}\beta_j X_j^\top \bar{X}_{\backslash j}\beta_{z,\backslash j}+\frac{1}{\sigma^2}\beta_j X_j^\top y-\beta_j^2(X^{\top} X)_{jj}/2\sigma^2)\, .
\end{align*}
Suppose half of the mass is concentrated on $e_1$ and the rest half evenly distributed among the remaining $2^p-1$ models. We start with $e_1+e_p$ (therefore $1$ false positive and no false negatives), for a choice of $\tau_1 > \tau_0$, let us take the first term $q\tau_0/(1-q)\tau_1 = o(1)$ since it is independent of $\beta$, the update reduces to
\[\frac{\pi(z_1=1\vert \beta,y,z_{-1})}{\pi(z_1=0\vert \beta,y,z_{-1})}\sim \exp(-\frac{1}{2}(\frac{1}{\tau_1^2}-\frac{1}{\tau_0^2})\beta_1^2+\frac{n}{2\sigma^2}\beta_1^2)\]
and for all other $j\neq 1$,
\[\frac{\pi(z_j=1\vert \beta,y,z_{-j})}{\pi(z_j=0\vert \beta,y,z_{-j})}\sim\exp(-\frac{1}{2}(1/\tau_1^2-1/\tau_0^2)\beta_j^2-\frac{1}{\sigma^2}\beta_j X_j^\top \bar{X}_{\backslash j}\beta_{z,\backslash j}+\frac{1}{\sigma^2}\beta_j X_j^\top X_1\beta_1^*-n\beta_j^2/2\sigma^2)\]
using $y=X_1\beta_1^*+\sigma \epsilon$ and assuming $(X^\top X)_{jj}=n$ is normalized. Additionally, we assume $X_1$ is orthogonal to all other columns. Under this assumption we have $\beta_1 \sim \beta_1^*$ and $\beta_{2,\dots,p}\sim 0$ after the first $\beta$ update (recall it amounts to regressing on the active components and setting the inactive ones to $\sim 0$). Therefore even though $z_1$ will stay $1$ and hence active with high probability, the rest of the $z_2,\dots, z_p$ will have almost equal probability of staying $0$ or $1$. The situation will likely repeat since $\beta_{2,\dots,p}\sim 0$ will remain. What we can conclude from this example is that the Gibbs sampler will witness (exponentially) long streaks of updates over the $2^{p}-1$ null models, followed by occupying the true model $e_1$ for equally long period of time and be very slow to move in between these two scenarios, since using \cref{prop:sparsify}
\begin{align*}
\frac{\pi(z_j=1\vert \beta,y,z_{-j})}{\pi(z_j=0\vert \beta,y,z_{-j})}&\sim \frac{q}{1-q}\frac{\tau_0}{\tau_1}\frac{\exp[-\frac{1}{2\sigma^2}\beta_1^{\top}(X_1^\top X_1+\frac{\sigma^2}{\tau_1^2}I)\beta_1+\frac{1}{\sigma^2}\beta_1^{*\top} X_1^\top X_1\beta_1]}{\exp(-\beta_j^2/2\tau_0^2)}\\
&\sim \exp[-\frac{1}{2\sigma^2}\beta_1^{\top}(X_1^\top X_1+\frac{\sigma^2}{\tau_1^2}I)\beta_1+\frac{1}{\sigma^2}\beta_1^{\top} X_1^\top X_1\beta_1]
\end{align*}
becomes very small for $j\neq 1$ when we have identified the true model $e_1$, which means $z_j$ will stay $0$ (i.e., inactive) with high probability. On the other hand, blocked updates that do not adopt a coordinate-by-coordinate strategy will switch between the two half of the time. 

We also point out that while the updates for Gibbs sampler is simple to implement, its mixing time is not immune to multi-modality. Consider the case when $X_1$ and $X_2$ are strongly correlated and the posterior puts half of the mass on the model consisting of these two variables only; and the other half evenly on the rest $2^{p-2}-1+2^{p-2}=2^{p-1}-1$ models (note that due to the correlation, either both $X_1$ and $X_2$ are included or not included, assuming the remaining $X_{\backslash\{1,2\}}$ are almost orthogonal to them). Such colinearity in the data shows up as coherence of $X$ (defined in \eqref{eqn:coherence} below) in the mixing time analysis of the Gibbs sampler. If one initializes with either $z_1=z_2=1$ or $z_1=z_2 = 0$, similar argument as above shows that the Gibbs update will be very slow moving in between these two cases (even though both make up non-negligible portion of the posterior $3/4$ vs. $1/4$) -- this is essentially because they form two separated peaks in the $z$-space. 

\subsubsection{Gibbs Mixing Guarantee for Posterior \eqref{eqn:quasi_posterior}}
\label{sec:mixing_gibbs}
We will loosely follow the approach taken in \cite{atchade2021approximate} which assumes that we can initialize from a model $z$ with no false negatives and at most $t$ false positives. The analysis is based on spectral gaps tailored to (finite) mixture of log-concave measures and allows one to restrict the study of spectral gaps to sets where most of the probability mass resides. Define for some $s\geq 0, \delta>0$
\begin{align} 
\mathcal{E}_s &:= \Big\{\pi(z\in\{0,1\}^p\colon z^*\subset z, \|z\|_0\leq \|z^*\|_0+s\vert y)\geq 1-\frac{4}{p^{\frac{\delta}{2}(s+1)}}\cap \pi(z^*\vert y)\geq 1/2 \label{eqn:event_1}\\
&\cap \max_{z^*\subset z, \|z\|_0\leq \|z^*\|_0+s} \max_{j\in[p], j\notin z} \; |\langle (I_n+\tau_1^2/\sigma^2 X_{z}X_{z}^\top)^{-1} X_j, \epsilon\rangle| \leq \sigma \sqrt{2(s+1)n\log(p)}\Big\} \label{eqn:event_2}
\end{align}
which is a high probability event over the randomness of the noise $\epsilon$ only ($X$ and $\beta^*$ are assumed to be fixed that satisfy certain conditions given below). Moreover, the design matrix $X$ has coherence for some integer $k\geq 1$,
\begin{equation}
\label{eqn:coherence}
\mathcal{C}(k) := \max_{\|z\|_0\leq k}\max_{j\neq i, j \notin z} |X_j^\top(I_n+\tau_1^2/\sigma^2 X_{z}X_{z}^\top)^{-1}X_i| \geq 0
\end{equation}
and restricted eigenvalue that entails $X^\top X$ is strongly convex in certain directions
\begin{equation}
\label{eqn:re}
\mathcal{\omega}(k) := \min_{z: \|z\|_0\leq k}\min_{\|v\|_2=1} \left\{v^\top X_{1-z}^\top(I_n+\tau_1^2/\sigma^2 X_{z}X_{z}^\top)^{-1}X_{1-z} v\colon v\in\mathbb{R}^{p-\|z\|_0}, \|v\|_0\leq k\right\} \geq 0\, .
\end{equation}
In general smaller $\mathcal{C}(k)$ and bigger $\omega(k)$ indicate better design, which in some sense capture the correlation between active and inactive components. Result of \cite{yang2016computational} suggests posterior concentration such as \eqref{eqn:event_1} alone isn't enough for efficient sampling if one allows arbitrary initialization, but these are the bare minimum and we will justify the posterior concentration property for the posterior \eqref{eqn:quasi_posterior} (i.e., the first two conditions  in $\mathcal{E}_s$) in \Cref{subsec:posterior_contract}. We additionally assume $\beta$-min condition for the true signal, i.e., 
\begin{equation}
\label{eqn:beta_min}
|\beta^*_{z^*,j}|\gtrsim \sigma\sqrt{\log(p)/n}, \quad \|\beta^*_{1-z^*}\|_2=0
\end{equation}
above the detection threshold for all active coordinates $j$, which is unavoidable if an initialization with no false negatives / contraction towards the true support is desired. 

Initializing from the support of Lasso can be a viable choice for warm-start. Even in the frequentist setup, it is popular to consider model selection with Lasso first, followed by regressing on the selected subset with (appropriately chosen) coordinated-weighted $\ell_1$-penalty ($\propto 1/|\hat{\beta}_{\text{init},j}|$) \`a la Adaptive Lasso \cite{buhlmann_geer_book}. Another possibility is to do a preliminary MCMC run on the posterior $\pi(z|y)$ first and hopefully identify the high-probability models. 

\begin{lemma}
\label{lem:warm-start}
The last condition in $\mathcal{E}_s$ holds with high probability and \eqref{eqn:coherence} is satisfied for $\mathcal{C}(k)\lesssim k^2\log(p)$, \eqref{eqn:re} is bounded away from $0$ for $k\sim n/\log(p)$ when e.g., the design matrix $X_{ij}\sim \mathcal{N}(0,1)$ for $n\gtrsim k\log(p)$. Moreover, with the above scaling of $\mathcal{C}(k)$, $\omega(k)$ and \eqref{eqn:beta_min}, Lasso has false positives bounded above by $\mathcal{O}(k)$, i.e., sparsity level of $\beta^*$, and no false negatives with high probability.
\end{lemma}
\begin{proof}
This is a modification of Lemma 8 and 9 of \cite{atchade2021approximate} so we will be brief. Since $\epsilon\sim \mathcal{N}(0,\sigma^2 I)$, \eqref{eqn:event_2} simply follows by observing that for $X_{ij}\sim \mathcal{N}(0,1)$,
\[\max_{z^*\subset z, \|z\|_0\leq \|z^*\|_0+s} \max_{j\in[p],z_j=0} \;  \|(I_n+\tau_1^2/\sigma^2 X_{z}X_{z}^\top)^{-1} X_j\|\leq \max_j \|X_j\|\lesssim \sqrt{n}\]
and the Gaussian deviation inequality. For the condition \eqref{eqn:coherence} and \eqref{eqn:re}, it is known when $n\gtrsim k\log(p)$, for Gaussian random matrix $\mathbb{P}(X\in \mathcal{H}) \gtrsim 1-1/p$, where 
\begin{align*}
\mathcal{H} := \big\{&X\in \mathbb{R}^{n\times p}\colon \|X_j\|_2\asymp \sqrt{n} \; \forall j\in[p],\, \max_{j\neq i}|\langle X_j,X_i\rangle |\lesssim \sqrt{n\log (p)}, \\
&\min_{\|v\|_0\leq k, \|v\|_2=1} v^\top (X^\top X)v \gtrsim n\big\}
\end{align*}
therefore we condition on the event $\mathcal{H}$ for the rest of the argument. Now Woodbury's identity and Cauchy Schwarz together with $\mathcal{H}$ give for $j\neq i$,
\begin{align*}
&|X_j^\top(I_n+\tau_1^2/\sigma^2 X_{z}X_{z}^\top)^{-1}X_i|=|X_j^\top X_i-X_j^\top X_{z}(\frac{\sigma^2}{\tau_1^2}I+X_{z}^\top X_{z})^{-1}X_{z}^\top X_i|\\
&\leq |X_j^\top X_i|+\sqrt{X_j^\top X_z(\frac{\sigma^2}{\tau_1^2}I+X_{z}^\top X_{z})^{-1}X_z^\top X_j}\sqrt{X_i^\top X_z(\frac{\sigma^2}{\tau_1^2}I+X_{z}^\top X_{z})^{-1}X_z^\top X_i}\\
&\lesssim \sqrt{n\log (p)}+\frac{1}{n} \|X_j^\top X_{z}\| \|X_{z}^\top X_i\|\\
&\lesssim \sqrt{n\log(p)}+\frac{k\sqrt{n\log(p)}(k\sqrt{n\log(p)}+n)}{n}\\
&\lesssim k^2 \log(p)
\end{align*}
for $X_j\notin X_{z}$ and $\|z\|_0=k$ and we used $n\gtrsim k\log(p)$. Similarly, for $\|z\|_0\leq k$ and $\text{supp}(v)\subset 1-z, \|v\|_0 \leq k$, on event $\mathcal{H}$, we have
\begin{align*}
v^\top X_{1-z}^\top(I_n+\tau_1^2/\sigma^2 X_{z}X_{z}^\top)^{-1}X_{1-z} v &= \| X_{1-z} v\|^2 - v^\top X_{1-z}^\top X_{z}(\frac{\sigma^2}{\tau_1^2}I+X_{z}^\top X_{z})^{-1}X_{z}^\top X_{1-z} v\\
&\gtrsim n\|v\|^2-\frac{\|X_{z}^\top X_{1-z} v\|^2}{n}\\
&\gtrsim n\|v\|^2 - \frac{k n\log(p)}{n}\|v\|^2 > 0
\end{align*}
for $n\gtrsim k\log(p)$. The warm start guarantee of Lasso for Gaussian design under $\beta$-min condition follows from classical results on support recovery \cite{buhlmann_geer_book}.  
\end{proof}

We will analyze a blocked variant of Gibbs with lazy updates (it is well-known that lazy version of the Markov chain only slows down the convergence by a constant factor). To implement, at step $k$ we perform the following updates. 
\begin{algorithm}[H]
\caption{Blocked Gibbs sampler for posterior \eqref{eqn:quasi_posterior}}
\label{alg:gibbs}
\begin{algorithmic}
\REQUIRE{$\beta_k\in \mathbb{R}^p, z_k\in \{0,1\}^p$}
\STATE{Sample $o\sim\text{Bern}(1/2)$}
\IF{$o=1$} 
\STATE{($\beta_{k+1},z_{k+1}) \leftarrow (\beta_k,z_k)$}
\ELSIF{$o=0$}
\STATE{Draw $\beta_{k+1}$ as in \eqref{eqn: beta_update} using \cref{alg:sample_gaussian} for the active part, the inactive part corresponding to $z_k[j]=0$ can be drawn independently}
\STATE{Sample $z_{k+1}^{1,2,3,\dots,p}\vert \beta_{k+1} \sim \pi(z^1|\beta_{k+1},y)\pi(z^2|z^1,\beta_{k+1},y)\pi(z^3|z^1,z^2,\beta_{k+1},y)\cdots$ where some marginalizations need to be done analytically, which is possible in this case of linear model with Gaussian slab}
\ENDIF
\RETURN $\beta_{k+1},z_{k+1}$
\end{algorithmic}
\end{algorithm}

%
Written mathematically, the Markov transition kernel takes the form
\[K(\beta_k,\beta_{k+1})=\sum_{z_{k+1}\in \{0,1\}^p} \pi(z_{k+1}|\beta_k, y)\left(\frac{1}{2}\delta_{\beta_k}(\beta_{k+1})+\frac{1}{2}\pi(\beta_{k+1}|z_{k+1},y)\right)\, .\]
\begin{remark}
The sampling of the $z_{k+1}|\beta_{k+1}$ step in \cref{alg:gibbs} is not particularly cheap, but our focus is on the mixing property of the Markov chain, and in light of the discussion in \Cref{sec:practical}, blocked updates as studied here only give a stronger guarantee in terms of mixing (there could generally be more bottlenecks in the chain).
\end{remark}
We preface with a lemma before stating our main result for the algorithm above. 
\begin{lemma}
\label{lem:prep}
The relative density for two models $\frac{\pi(z_2\vert y)}{\pi(z_1\vert y)}$ where $z_1\subset z_2$ can be shown to be as \eqref{eqn:relative_density_1}-\eqref{eqn:relative_density_2}, and given tolerance $\zeta_0\in (0,1)$, assuming $q/(1-q)\sim 1/p^{\delta+1}$ for some $\delta>0$, $\|X_j\|_2^2=n\; \forall j\in[p]$, we have
\[\|\pi_0 K^k-\pi(\beta|y)\|_{\text{TV}}\leq 2p^{(\delta+1)t} (1+\frac{\tau_1^2\cdot tn}{\sigma^2})^{t/2}(1-\text{SpecGap}_\zeta(K))^{k/2}+\zeta_0/\sqrt{2}\]
for $\zeta=\frac{\zeta_0^2}{8}p^{-2(\delta+1)t} (1+\frac{\tau_1^2\cdot tn}{\sigma^2})^{-t}$ if we initialize with $t$ false-positives and no false negatives. 
\end{lemma}
\begin{proof}
The posterior marginal over finite state space $z\in \{0,1\}^p$ after integrating out $\beta(z) =[\bar{\beta} \; \bar{\beta}_c]$ is (this is a special feature of conjugate priors)
\begin{align*}
&\pi(z\vert y) \propto q^{\|z\|_0}(1-q)^{p-\|z\|_0} \times \\
&\frac{\tau_0^{\|z\|_0-p}}{\tau_1^{\|z\|_0}}\int_{\mathbb{R}^p} \exp\left(-\frac{1}{2\sigma^2}(\bar{\beta}^\top \bar{X}^{\top}\bar{X}\bar{\beta}-2\bar{\beta}^\top \bar{X}^\top y)-\bar{\beta}^\top D(\frac{1}{2\tau_1^2})\bar{\beta}-\bar{\beta}_c^\top D(\frac{1}{2\tau_0^2})\bar{\beta}_c\right)d\beta\\
&\propto q^{\|z\|_0}(1-q)^{p-\|z\|_0}(\frac{\tau_0}{\tau_1})^{\|z\|_0} (\tau_0^2)^{(p-\|z\|_0)/2} \frac{\exp(\frac{1}{2\sigma^4}y^\top\bar{X}(\frac{1}{\sigma^2}\bar{X}^\top\bar{X}+1/\tau_1^2 \cdot I)^{-1}\bar{X}^\top y)}{\sqrt{\det(\frac{1}{\sigma^2}\bar{X}^\top\bar{X}+1/\tau_1^2\cdot I)}}\\
&\propto q^{\|z\|_0}(1-q)^{p-\|z\|_0}(\frac{\tau_0}{\tau_1})^{\|z\|_0} (\tau_0^2)^{(p-\|z\|_0)/2} \frac{\exp(\frac{1}{2\sigma^4}y^\top\bar{X}(\frac{1}{\sigma^2}\bar{X}^\top\bar{X}+1/\tau_1^2 \cdot I)^{-1}\bar{X}^\top y)}{\sqrt{\det(I_n+\tau_1^2/\sigma^2 \bar{X}\bar{X}^\top)}} (\tau_1^2)^{\|z\|_0/2}\\
&\propto (\frac{q}{1-q})^{\|z\|_0}(\frac{\tau_0}{\tau_1})^{\|z\|_0}(\frac{\tau_1}{\tau_0})^{\|z\|_0}\frac{\exp(\frac{1}{2\sigma^4}y^\top\bar{X}(\tau_1^2\cdot I-\tau_1^4\bar{X}^\top(\sigma^2I+\tau_1^2\bar{X}\bar{X}^\top)^{-1}\bar{X})\bar{X}^\top y)}{\sqrt{\det(I_n+\tau_1^2/\sigma^2 \bar{X}\bar{X}^\top)}} \\
&\propto (\frac{q}{1-q})^{\|z\|_0} \frac{\exp(-\frac{1}{2\sigma^2}y^\top (I_n+\tau_1^2/\sigma^2 X_{z}X_{z}^\top)^{-1}y)}{\sqrt{\det(I_n+\tau_1^2/\sigma^2 X_{z}X_{z}^\top)}}
\end{align*}
where we used (1) Gaussian integral $\int_{\mathbb{R}^k}\exp(-\frac{1}{2}x^\top\Sigma^{-1}x)dx=(2\pi)^{k/2}\det(\Sigma)^{1/2}$ and completion of squares; (2) matrix determinant lemma $\det(A+UV^\top)=\det(A)\det(I+V^\top A^{-1}U)$; (3) Woodbury identity $(A+UCV)^{-1}=A^{-1}-A^{-1}U(C^{-1}+VA^{-1}U)^{-1}VA^{-1}$ and the fact that
\begin{align*}
&y^\top (\bar{X}\bar{X}^\top-\tau_1^2\bar{X}\bar{X}^\top(\sigma^2I+\tau_1^2\bar{X}\bar{X}^\top)^{-1}\bar{X}\bar{X}^\top)y\\
&=\sigma^2 y^\top\bar{X}\bar{X}^\top(\sigma^2I+\tau_1^2\bar{X}\bar{X}^\top)^{-1}y=\sigma^2 y^\top \bar{X} (\tau_1^2 \bar{X}^\top \bar{X}+\sigma^2 I)^{-1} \bar{X}^\top y\\
&=\frac{\sigma^2}{\tau_1^2}y^\top (I-\sigma^2(\sigma^2I+\tau_1^2\bar{X}\bar{X}^\top)^{-1}) y
\end{align*}
for the last step. Now if we want to look at the change in posterior for two models $z_1$ and $z_2$ where $z_1\subset z_2$, since both numerator and denominator involve
\[I_n+\tau_1^2/\sigma^2 X_{z_2}X_{z_2}^\top=I_n+\tau_1^2/\sigma^2 X_{z_1}X_{z_1}^\top+\tau_1^2/\sigma^2\sum_{j\colon z_{1,j}=0,z_{2,j}=1} X_jX_j^\top\, ,\]
matrix determinant lemma and Woodbury identity will again let us compute the ratio
\begin{align}
&\frac{\pi(z_2\vert y)}{\pi(z_1\vert y)} = (\frac{q}{1-q})^{\|z_2\|_0-\|z_1\|_0}\times \frac{1}{\sqrt{\det(I+\frac{\tau_1^2}{\sigma^2}X_{z_2-z_1}^\top A^{-1}X_{z_2-z_1})}} \label{eqn:relative_density_1}\\
&\times \exp\left(\frac{1}{2\sigma^2}y^\top A^{-1}X_{z_2-z_1}(\frac{\sigma^2}{\tau_1^2}I+X_{z_2-z_1}^\top A^{-1}X_{z_2-z_1})^{-1}X_{z_2-z_1}^\top A^{-1}y\right)\, , \label{eqn:relative_density_2}
\end{align}
where $A=I_n+\tau_1^2/\sigma^2 X_{z_1}X_{z_1}^\top\succeq I_n$ and $X_{z_2-z_1}$ denotes columns of $X$ for which $z_{1,j}=0$ and $z_{2,j}=1$. Let us denote the initial model as $z_0$, and define
\[f_0(\beta) := \frac{\pi(\beta|z_0,y)}{\pi(\beta|y)}\leq \frac{1}{\pi(z_0|y)}\leq \frac{2\pi(z^*|y)}{\pi(z_0|y)}\]
since $\pi(z^*|y) \geq 1/2$ on the event $\mathcal{E}_s$. This implies using \eqref{eqn:relative_density_1}-\eqref{eqn:relative_density_2} that since $z^*\subset z_0$, denoting the number of initial false positives as $t$, and using the assumptions
\begin{align*}
\|f_0\|_{\pi,\infty} &:= \text{ess}\sup|f_0(\beta)|\; \text{w.r.t}\; \pi(d\beta) \\
&\leq 2p^{(\delta+1)t} \sqrt{\det(I_t+\frac{\tau_1^2}{\sigma^2}X_{z_0-z^*}^\top A^{-1}X_{z_0-z^*})} \leq 2p^{(\delta+1)t} (1+\frac{\tau_1^2\cdot tn}{\sigma^2})^{t/2}\, .
\end{align*}
Using Lemma 1 from \cite{atchade2021approximate} we have for all iterations $k\geq 1$ and initial $\pi_0(d\beta)=\pi(\beta|z_0,y)$,
\begin{align*}
&\|\pi_0 K^k-\pi(\beta|y)\|_{\text{TV}}^2 \\
&\leq \max\left\{\int |f_0(\beta)-\int f_0(\beta)\pi(d\beta)|^2 \pi(d\beta),\zeta\|f_0\|_{\pi,\infty}^2\right\}(1-\text{SpecGap}_\zeta(K))^k+\zeta \|f_0\|_{\pi,\infty}^2\\
&\leq \|f_0\|_{\pi,\infty}^2(1-\text{SpecGap}_\zeta(K))^k+\zeta_0^2/2
\end{align*}
if setting $\zeta=\frac{\zeta_0^2}{8}p^{-2(\delta+1)t} (1+\frac{\tau_1^2\cdot tn}{\sigma^2})^{-t}$ for some $\zeta_0\in(0,1)$ the desired accuracy. 
\end{proof}

With these preparations, it only remains to bound the approximate spectral gap from \cref{lem:prep} to conclude, for which we leverage the framework developed in \cite{atchade2021approximate}. At a high level it states that if when constrained on a subset of models $\bar{z}$ where the posterior mass concentrates, the marginal densities $\pi(\beta|z_1,y),\pi(\beta|z_2,y)$ overlap sufficiently for $z_1,z_2$ on this set that are somewhat close to each other, $\zeta$-spectral gap can be much larger than the classically defined spectral gap without such a restriction (hence tighter resulting bounds).
\begin{hypothesis}
\label{assume:scaling}
We assume for some $\delta>0$, $q/(1-q)\sim 1/p^{\delta+1},\tau_1\sim \sigma p/\sqrt{n}, \tau_0\sim \sigma/\sqrt{n},\|X_j\|_2^2=n$ for all $j\in[p]$. Throughout the paper we consider $q\in (0,1)$ to be fixed, i.e., non-data-adaptive as opposed to empirical Bayes approaches in the literature.
\end{hypothesis}
\begin{proposition}[Convergence Rate for Gibbs Sampler]
\label{thm:gibbs}
Under the event $\mathcal{E}_s$, condition \eqref{eqn:coherence},\eqref{eqn:re},\eqref{eqn:beta_min}, \cref{assume:scaling} and warm start with number of false positives $t\geq 0$ bounded above as 
\[(\frac{1}{p})^{2(1+\delta)t}(\frac{1}{1+tp^2})^t\geq \frac{20}{p^{\frac{\delta}{2} (s+1)}\zeta_0^2} \,,\]
after 
\[k\gtrsim (s+1) p^{(1+\delta)t} (1+tp^2)^{t/2}\exp\left(\frac{ns^2}{\sigma^2}\eta^2+\frac{2\sqrt{n\log(p)}}{\sigma}\eta+\frac{n}{2\sigma^2}\eta^2\right)\log\left(\frac{1}{\zeta_0}\right)\]
steps of \cref{alg:gibbs}, we have $\|\pi_0 K^k-\pi(\beta|y)\|_{\text{TV}} \leq \zeta_0$. In particular, if $s=0, \delta=1$, the iteration complexity is 
\[k\gtrsim p^{2t} (1+tp^2)^{t/2}\exp\left(\left(\frac{\sqrt{n}}{\sqrt{\omega(k)}} \vee \frac{n^2}{\omega^2(k)}\right) \log(p)+\left(\frac{k\mathcal{C}(k)}{\omega(k)} \vee \frac{k^2\mathcal{C}^2(k)}{\omega^2(k)}\right) \log(p)\right)\log\left(\frac{1}{\zeta_0}\right)\,.\]
Each iteration implemented with \cref{alg:sample_gaussian} costs at least $\mathcal{O}(\max\{n^2k,n^3\})$.  
\end{proposition}
\begin{proof}
On the event $\mathcal{E}_s$, we have that the posterior puts at least $1-\frac{\zeta}{10}=1-\frac{\zeta_0^2}{80}p^{-2(\delta+1)t} (1+\frac{\tau_1^2\cdot tn}{\sigma^2})^{-t}$ fraction of the mass on the set 
\begin{equation}
\label{eqn:posterior}
\pi(z\in\{0,1\}^p\colon z^*\subset z, \|z\|_0\leq \|z^*\|_0+s\vert y)\geq 1-\frac{4}{p^{\frac{\delta}{2}(s+1)}}
\end{equation}
if picking the initial false positives $t$ small enough such that given $s\geq 0,\zeta_0\in (0,1),\delta>0$
\begin{equation}
\label{eqn:warm-start}
(\frac{1}{p})^{2(1+\delta)t}(\frac{1}{1+t\tau_1^2 n/\sigma^2})^t\geq \frac{20}{p^{\frac{\delta}{2} (s+1)}\zeta_0^2} ,
\end{equation}
so the statement of Theorem 3 from \cite{atchade2021approximate} applies (picking $m=\infty, B_i=\mathbb{R}^p$) and we need to find $\kappa>0$ such that $\forall z_1,z_2$ belonging to the set \eqref{eqn:posterior} $=I_0$ that differs in 1 element (so both $z_1, z_2$ have at most $s$ false positives),
\begin{equation}
\label{eqn:kappa}
\int_{\mathbb{R}^p} \min\{\pi(\beta|z_1,y),\pi(\beta|z_2,y)\} d\beta\geq \kappa \,.
\end{equation}
Suppose w.l.o.g $z_1\subset z_2$ where $z_{1,j}=0$ and $z_{2,j}=1$, using \cref{lem:prep}, \eqref{eqn:quasi_posterior} we have for $A=I_n+\tau_1^2/\sigma^2 X_{z_1}X_{z_1}^\top\succeq I_n$ and under \cref{assume:scaling},
\begin{align*}
&\frac{\pi(\beta|z_1,y)}{\pi(\beta|z_2,y)} = \frac{\pi(\beta,z_1|y)}{\pi(z_1|y)} \frac{\pi(z_2|y)}{\pi(\beta,z_2|y)}\\
&= \frac{q}{1-q}\frac{1}{\sqrt{1+\frac{\tau_1^2}{\sigma^2}X_{j}^\top A^{-1}X_{j}}}\exp\left(\frac{1}{2\sigma^2}\frac{(y^\top A^{-1}X_{j})^2}{\frac{\sigma^2}{\tau_1^2}+X_{j}^\top A^{-1}X_{j}}\right)\frac{1-q}{q}\frac{\tau_1}{\tau_0}\exp(\frac{\beta_j^2}{2}(1/\tau_1^2-1/\tau_0^2))\\
&\times \exp\left(-\frac{1}{\sigma^2}y^\top X_j \beta_j+\frac{n}{2\sigma^2}\beta_j^2+\frac{\beta_jX_j^\top X_{z_1}\beta_{z_1}}{\sigma^2} \right) \\
&\geq \frac{p}{\sqrt{1+\frac{\tau_1^2}{\sigma^2}n}}\exp\left(\frac{1}{2\sigma^2}\frac{(y^\top A^{-1}X_{j})^2}{\frac{\sigma^2}{\tau_1^2}+X_{j}^\top A^{-1}X_{j}}-\frac{\beta_j^2}{2}\frac{1}{\tau_0^2}+\frac{\beta_jX_j^\top (X_{z_1}\beta_{z_1}-y)}{\sigma^2}\right)\\
&\geq \exp\left(-\frac{n}{2\sigma^2}\beta_j^2-\frac{|\beta_jX_j^\top (X_{z_1}\beta_{z_1}-y)|}{\sigma^2} \right)\, .
\end{align*}
Since both $z_1,z_2$ contain $z^*$, it must be the case $j\notin z^*=\text{supp}(\beta^*)$ with $|\text{supp}(\beta^*)|=k$, and as $X_j$ is not part of $z_1$, under event $\mathcal{E}_s$ and \cref{assume:scaling},
\begin{align*}
&\frac{1}{\sigma^2}|\beta_j X_j^\top (X\beta^*+\epsilon- X_{z_1}\beta_{z_1})| \\
&\leq \frac{1}{\sigma^2} (|\beta_jX_j^\top X_s\beta_s|+|\beta_jX_j^\top\epsilon|) \\
&\leq \frac{ns}{\sigma^2} |\beta_j|\|\beta_s\|_1 +\frac{2}{\sigma}|\beta_j|\sqrt{n\log(p)}
\end{align*}
where $X_s$ is the $n$-by-at-most-$s$ matrix composed of columns of $X$ that are in the $z_1$ model (and therefore $z_2$) but not in $z^*$ (these are false positives). 
Now take any $j$ that is not in $z^*$ but is in $z_2$, we know from \cref{prop:sparsify} the marginal distribution $\pi(\beta_j|z_2,y)$ is Gaussian with absolute value of the mean bounded as (using the definition of \eqref{eqn:coherence},\eqref{eqn:re})
\begin{align*}
&\left| e_1^\top \begin{bmatrix}
X_j^\top X_j+\sigma^2/\tau_1^2 & X_j^\top X_{z_2\backslash j}\\
X_{z_2\backslash j}^\top X_j & X_{z_2\backslash j}^\top X_{z_2\backslash j}+\sigma^2/\tau_1^2 \cdot I
\end{bmatrix}^{-1}\begin{bmatrix}
X_j^\top y\\
X_{z_2\backslash j}^\top y
\end{bmatrix}\right|\\
&= \left|\frac{X_j^\top y-X_j^\top X_{z_2\backslash j}(X_{z_2\backslash j}^\top X_{z_2\backslash j}+\sigma^2/\tau_1^2\cdot I)^{-1}X_{z_2\backslash j}^\top y}{\frac{\sigma^2}{\tau_1^2}+X_j^\top(I+\tau_1^2/\sigma^2\cdot X_{z_2\backslash j}X_{z_2\backslash j}^\top)^{-1}X_j}\right|\\
&=\left| \frac{X_j^\top (I+\tau_1^2/\sigma^2\cdot X_{z_2\backslash j}X_{z_2\backslash j}^\top)^{-1} (X\beta^*+\epsilon)}{\frac{\sigma^2}{\tau_1^2}+X_j^\top(I+\tau_1^2/\sigma^2\cdot X_{z_2\backslash j}X_{z_2\backslash j}^\top)^{-1}X_j}\right|\leq \frac{\sigma \sqrt{2(s+1)n\log(p)}+\|\beta^*\|_1\mathcal{C}(k+s)}{\omega(k+s)}
\end{align*}
by H\"older and triangle inequality and variance 
\begin{align*} 
&\sigma^2[(X_{z_2}^\top X_{z_2}+\sigma^2/\tau_1^2 I)^{-1}]_{jj}\\
&=\frac{\sigma^2}{X_j^\top X_j+\frac{\sigma^2}{\tau_1^2}-X_j^\top X_{z_2\backslash j}(X_{z_2\backslash j}^\top X_{z_2\backslash j}+\sigma^2/\tau_1^2 \cdot I)^{-1}X_{z_2\backslash j}^\top X_j}\\
&=\frac{\sigma^2}{\frac{\sigma^2}{\tau_1^2}+X_j^\top(I+\tau_1^2/\sigma^2\cdot X_{z_2\backslash j}X_{z_2\backslash j}^\top)^{-1}X_j} \leq \frac{\sigma^2}{\omega(k+s)}
\end{align*}
where we used matrix block inversion and Woodbury identity. Noting that since these two expressions are independent of the choice of $j$, which in particular means that the upper bound holds for any such $j$, we can write $\beta_j=\mu_j+\sigma_j z$ for $z\sim\mathcal{N}(0,1)$, and
\begin{align*}
&\int_{\mathbb{R}^p} \min\{\pi(\beta|z_1,y),\pi(\beta|z_2,y)\} d\beta=\mathbb{E}_{\pi(\beta|z_2,y)}\left[\min\{\frac{\pi(\beta|z_1,y)}{\pi(\beta|z_2,y)},1\}\right] \\
&\geq \mathbb{E}_{\beta_s}\left[\mathbb{E}_{\beta_j}\left[\exp\left(-\frac{n}{2\sigma^2}\beta_j^2-\frac{ns}{\sigma^2} |\beta_j|\|\beta_s\|_1 -\frac{2}{\sigma}|\beta_j|\sqrt{n\log(p)} \right)\vert \beta_s \right]\right]\\
&\geq \frac{1}{2}\mathbb{E}_{\beta_s}\left[\exp\left(-\frac{ns}{\sigma^2}(|u_j|+\sigma_j)\|\beta_s\|_1-\frac{2\sqrt{n\log(p)}}{\sigma}(|u_j|+\sigma_j)-\frac{n}{2\sigma^2}(|u_j|+\sigma_j)^2\right)\right]\\
&\geq \frac{1}{2}\exp\left(-\frac{ns}{\sigma^2}(|u_j|+\sigma_j)\mathbb{E}[\|\beta_s\|_1]-\frac{2\sqrt{n\log(p)}}{\sigma}(|u_j|+\sigma_j)-\frac{n}{2\sigma^2}(|u_j|+\sigma_j)^2\right)\\
&\geq \frac{1}{2}\exp\left(-\frac{ns^2}{\sigma^2}\eta^2-\frac{2\sqrt{n\log(p)}}{\sigma}\eta-\frac{n}{2\sigma^2}\eta^2\right)
\end{align*}
where we used that (1) for any non-negative function $g$, $\mathbb{E}[g(z)]\geq \mathbb{P}(|z|\leq 1)\min_{z:|z|\leq 1} g(z)$; (2) Jensen's inequality; (3) for any coordinate $j$ of $\beta_s$, 
\begin{align*}
\mathbb{E}[|\beta_s[j]|]\leq \sqrt{\mathbb{E}[\beta_s[j]^2]} &\leq \sqrt{\frac{\sigma^2}{\omega(k+s)}+\left(\frac{\sigma \sqrt{2(s+1)n\log(p)}+\|\beta^*\|_1\mathcal{C}(k+s)}{\omega(k+s)}\right)^2}\\
&\leq \frac{\sigma}{\sqrt{\omega(k+s)}}+\frac{\sigma \sqrt{2(s+1)n\log(p)}+\|\beta^*\|_1\mathcal{C}(k+s)}{\omega(k+s)} =: \eta\, ,
\end{align*}
(4) it holds that $|\mu_j|+\sigma_j\leq \eta$.
Therefore one can invoke Theorem 3 with 
\[\kappa= \frac{1}{2}\exp\left(-\frac{ns^2}{\sigma^2}\eta^2-\frac{2\sqrt{n\log(p)}}{\sigma}\eta-\frac{n}{2\sigma^2}\eta^2\right)\]
for \eqref{eqn:kappa}. Using that the diameter of the graph constructed on $I_0$ (where $z_1,z_2\in I_0$ differing in 1 element) is bounded above by $2s$, we reach
\begin{align*}
\text{SpecGap}_\zeta(K) &\geq \frac{\kappa}{4s}\min_{z: z^*\subset z, \|z\|_0\leq \|z^*\|_0+s} \pi(z|y)\\
&\gtrsim \frac{1}{s\zeta_0 p^{\delta (s+1)/4}}\exp\left(-\frac{ns^2}{\sigma^2}\eta^2-\frac{2\sqrt{n\log(p)}}{\sigma}\eta-\frac{n}{2\sigma^2}\eta^2\right)
\end{align*}
where we used the relative ratio from \cref{lem:prep}, for any $z$ with at most $s$ false positives,
\[\pi(z|y)\geq \frac{1}{2}\frac{\pi(z|y)}{\pi(z^*|y)}\geq \frac{1}{2p^{s(\delta+1)}}(1+\frac{\tau_1^2 ns}{\sigma^2})^{-s/2}\geq \frac{1}{2p^{s(\delta+1)}}(1+p^2 s)^{-s/2}\gtrsim p^{-\frac{\delta(s+1)}{4}}\zeta_0^{-1}\, .\]
Putting together with \cref{lem:prep,eqn:warm-start} now yields $\|\pi_0 K^k-\pi(\beta|y)\|_{\text{TV}} \leq \zeta_0$ when 
\begin{align*}
k&\gtrsim (s+1)\zeta_0 p^{\frac{\delta(s+1)}{4}}\exp\left(\frac{ns^2}{\sigma^2}\eta^2+\frac{2\sqrt{n\log(p)}}{\sigma}\eta+\frac{n}{2\sigma^2}\eta^2\right) \log\left(\frac{p^{(\delta+1)t}(1+p^2 t)^{t/2}}{\zeta_0}\right)\\
&\gtrsim  (s+1) p^{(1+\delta)t} (1+tp^2)^{t/2}\exp\left(\frac{ns^2}{\sigma^2}\eta^2+\frac{2\sqrt{n\log(p)}}{\sigma}\eta+\frac{n}{2\sigma^2}\eta^2\right)\log(\frac{1}{\zeta_0})\,,
\end{align*}
where we hide a poly-logarithmic factor in $p$. In the case of $s=0, \delta=1$, the posterior puts most of the mass on $z^*$, and we have
\[k\gtrsim p^{2t} (1+tp^2)^{t/2}\exp\left(\left(\frac{\sqrt{n}}{\sqrt{\omega(k)}} \vee \frac{n^2}{\omega^2(k)}\right) \log(p)+\left(\frac{k\mathcal{C}(k)}{\omega(k)} \vee \frac{k^2\mathcal{C}^2(k)}{\omega^2(k)}\right) \log(p)\right)\log(\frac{1}{\zeta_0})\,,\]
where we used the separation condition on the signal \eqref{eqn:beta_min} to estimate $\|\beta^*\|_1 \geq k\sigma \sqrt{\log(p)/n}$.
\end{proof} 

\cref{thm:gibbs} therefore implies that warm-start (made possible by frequentist estimators) is one way of getting around the hardness result of \cite{yang2016computational}. Other than the less-than-ideal scaling with the number of false positives $t$ (which capture the bottleneck moving in between lower and higher density regions), we'd like to note the exponential dependence of the mixing time on the coherence $\mathcal{C}(k)$ and restricted eigenvalue parameter $\omega(k)$ of the design matrix $X$ -- these won't be present if not due to spectral gap considerations, and it shows up even with warm start.  

\subsection{Spike-and-Slab for Random Design}
We consider a slightly different task in this section where the goal is to sample from a posterior $\pi(\beta\vert y)$ of the following form: given $y$ and assume $X_{i,j}\sim\mathcal{N}(0,1)$ independently, 
\[
\pi(\beta|y) \propto \sum_{z\in \{0,1\}^p}  \int_{X\in \mathbb{R}^{n\times p}} \exp(-\frac{1}{2\sigma^2}\|y-X\beta\|_2^2)\, \mu_{\text{G}}(dX)\cdot \prod_{j=1}^p ((1-q)G_0(\beta_j))^{1-z_j}(qG_1(\beta_j))^{z_j}
\]
with spike-and-slab prior on the parameter $\beta\in\mathbb{R}^p$. This is closer to random design setup where $y=X\beta+\epsilon$ for both $X,y$ a random sample as opposed to just $y$, and one could be interested in the performance of $\hat{\beta}\sim \pi(\beta\vert y)$ on future pairs of $(X,y)$ from the same model. The Gaussian i.i.d entry assumption of course hardly holds in practice but it may serve as a good proxy for some class of design matrix. The posterior, shown below in \cref{lem:random_design}, is only a function of $y$ (therefore no expensive matrix inversion involved in the algorithm), and if $q=0$, the density only depends on the magnitude $\|\beta\|$ which means that it's rotationally invariant (i.e., equal probability over sphere of fixed radius). For $q\neq 0$, due to the combinatorial nature of the mixture it introduces challenge for high-dimensional sampling -- na\"ively it could be exponential in $p$. 
\begin{lemma} 
\label{lem:random_design}
The posterior with continuous Gaussian Spike-and-Slab prior under random design takes the form (with $\tau_1 \gg \tau_0$)
\begin{align*} 
&\pi(\beta,z | y)\propto\\
&\frac{\sigma^n}{(\|\beta\|^2+\sigma^2)^{n/2}}\exp(\|y\|^2\frac{ \|\beta\|^2}{2\sigma^4+2\sigma^2\|\beta\|^2})\prod_{j=1}^p [\frac{1-q}{\tau_0}\exp(-\beta_j^2/2\tau_0^2)]^{1-z_j}\cdot[\frac{q}{\tau_1}\exp(-\beta_j^2/2\tau_1^2)]^{z_j}
\end{align*}
which is non-log-concave, but it is amenable to Gibbs updates (that is known to be reversible).
\end{lemma}
\begin{proof}
We calculate, since the entries of $X$ are assumed to be independent,
\begin{align*}
&\int_X \exp\left(-\frac{1}{2\sigma^2}\|y-X\beta\|_2^2\right)\, \mu_{\text{G}}(dX)\\
&\propto\prod_{i=1}^n \left[\int_{\mathbb{R}^p} \exp(\frac{1}{\sigma^2}y_ix_i^\top \beta-\frac{1}{2\sigma^2}\beta^\top x_ix_i^\top\beta-\frac{1}{2}\|x_i\|_2^2) \, dx_i\right]\\
&=\prod_{i=1}^n\left[\int_{\mathbb{R}^p} \exp(\frac{1}{\sigma^2}y_i\beta^\top x_i-\frac{1}{2\sigma^2}x_i^\top(\beta\beta^\top+\sigma^2 I)x_i) \, dx_i\right]\\
&= \prod_{i=1}^n \exp(\frac{1}{2\sigma^2}y_i^2 \beta^\top(\beta\beta^\top+\sigma^2 I)^{-1}\beta) \times\\
&\quad\quad\quad \int_{\mathbb{R}^p} \exp(-\frac{1}{2\sigma^2}[x_i-y_i(\beta\beta^\top+\sigma^2I)^{-1}\beta]^\top(\beta\beta^\top+\sigma^2I)[x_i-y_i(\beta\beta^\top+\sigma^2I)^{-1}\beta]) dx_i\\
&\propto \prod_{i=1}^n \exp\left(\frac{y_i^2}{2\sigma^2}\cdot \frac{1/\sigma^2\|\beta\|_2^2}{1+1/\sigma^2\|\beta\|_2^2}\right)\sqrt{\det(\sigma^2(\beta\beta^\top+\sigma^2 I)^{-1})}\\
&\propto \prod_{i=1}^n  \exp\left(\frac{y_i^2}{2\sigma^2}\cdot \frac{1/\sigma^2\cdot \|\beta\|_2^2}{1+1/\sigma^2\cdot \|\beta\|_2^2}\right)\frac{\sigma^p}{\sqrt{\|\beta\|^2+\sigma^2}\sigma^{(p-1)}}
\end{align*}
where we used Gaussian integral and the Sherman–Morrison formula, as claimed. 
Gibbs update alternate between
\begin{align}
\pi (\beta\vert z,y) &\propto \frac{\sigma^n}{(\|\beta\|^2+\sigma^2)^{n/2}}\exp(\|y\|^2\frac{ \|\beta\|^2}{2\sigma^4+2\sigma^2\|\beta\|^2})\mathcal{N}(\beta;0,D^{-1}) \nonumber\\
&\propto \frac{\sigma^n}{(\|\beta\|^2+\sigma^2)^{n/2}}\exp(\|y\|^2\frac{ \|\beta\|^2}{2\sigma^4+2\sigma^2\|\beta\|^2}-\frac{1}{2}\beta^\top D(z)\beta) \label{eqn:gibbs_1}
\end{align}
where $D(z):=\text{Diag}(z\tau_1^{-2}+(1_p-z)\tau_0^{-2})$ is a positive-definite diagonal matrix, and
\begin{align}
\pi(z\vert \beta,y) &\propto \prod_{j=1}^p (q\mathcal{N}(\beta_j;0,\tau_1^2))^{z_j}\cdot((1-q)\mathcal{N}(\beta_j;0,\tau_0^2))^{1-z_j} \nonumber\\
&\sim \prod_{j=1}^p \text{Bern}\left(z_j;\frac{q\mathcal{N}(\beta_j;0,\tau_1^2)}{q\mathcal{N}(\beta_j;0,\tau_1^2)+(1-q)\mathcal{N}(\beta_j;0,\tau_0^2)}\right) \nonumber\\
&= \prod_{j=1}^p Q_j^{z_j}(1-Q_j)^{1-z_j}\quad\text{for}\quad Q_j = \frac{1}{1+\frac{1-q}{q}\frac{\tau_1}{\tau_0}\exp(\frac{1}{2}(1/\tau_1^2-1/\tau_0^2)\beta_j^2)}
\end{align}
is product of independent Bernoulli's that can be sampled in parallel.


The marginal over $\beta$ is not log-concave therefore standard off-the-shelf sampler (e.g., Langevin, HMC etc.) doesn't come with efficiency guarantee, even in continuous time. To see this, we simply calculate the Hessian for the negative log density in \eqref{eqn:gibbs_1}, 
\[-\nabla \log \pi(\beta|z,y)=\frac{n}{\|\beta\|^2+\sigma^2}\beta-\|y\|^2(\frac{1}{\sigma^4+\sigma^2\|\beta\|^2}-\frac{\|\beta\|^2}{(\sigma^3+\sigma\|\beta\|^2)^2})\beta+D(z)\beta\]
for some $\frac{1}{\tau_1^2}I \preceq D(z) \preceq \frac{1}{\tau_0^2} I$. And
\begin{align*}
-\nabla^2\log \pi(\beta|z,y) &= \left(\frac{n}{\|\beta\|^2+\sigma^2}-\frac{\|y\|^2}{\sigma^4+\sigma^2\|\beta\|^2}+\frac{\|y\|^2\|\beta\|^2}{(\sigma^3+\sigma\|\beta\|^2)^2}\right)I+D(z)\\
&-\left(\frac{2\sigma\|\beta\|^2\|y\|^2-2\sigma^3\|y\|^2}{(\sigma^3+\sigma\|\beta\|^2)^3}-\frac{2\|y\|^2}{(\sigma^3+\sigma\|\beta\|^2)^2}+\frac{2n}{(\|\beta\|^2+\sigma^2)^2}\right)\beta\beta^\top
\end{align*}
which we can see is not always positive semi-definite on the entire domain of $\beta$, e.g., for a counter-example one could consider $\sigma^2 \ll \|\beta\|^2, \|y\|^2/\sigma^2 \ll n, \tau_0^2 \gg \|\beta\|^2/n$. Therefore the posterior $\pi(\beta|y)$ in this case is in fact a mixture of \emph{non}-log-concave measures, unlike the fixed design case in \cref{prop:sparsify}.
\end{proof}

\subsubsection{Inner Step Implementation of \eqref{eqn:gibbs_1} for Gibbs}
As it turns out target that has a density with respect to the Gaussian measure is somewhat easy to sample from. Consider the problem of sampling from the \emph{un-normalized} density $\pi(x)\propto f(x)\mathcal{N}(0,\gamma I)$ for $f>0$, where one can think of the prior as being Gaussian, and is performing optimal transport from $\mathcal{N}(0,\gamma I)$ to $\pi$ in the space of probability measures. Schr\"odinger bridge admits closed-form expression as an SDE if starting at the origin at $t=0$. It is known from \cite{schrodinger, Raginsky} that
\[Q^{\pi}:=\arg\min_{Q\in \mathcal{M}^\pi} \text{KL}(Q\vert\vert P)\]
where $\mathcal{M}^{\pi}=\{Q:Q_0=\delta_0, Q_1=\pi\}$ the set of distributions with the two time marginals pinned at $t=0$ and $t=1$ end points and $P$ the reference Wiener measure associated with the process 
\[dX_t = \sqrt{\gamma} dW_t, \; X_0 \sim \delta_0\,,\]
is governed by an SDE with time-varying F\"ollmer drift (i.e., depends on both $X_t$ and $t$, unlike Langevin):
\begin{equation}
\label{eqn:sde_schrodinger}
dX_t = \nabla_X \log \mathbb{E}_Z [f(X_t+\sqrt{1-t}Z)]dt+ \sqrt{\gamma}  dW_t, \quad X_0=0,\, t\in [0,1]
\end{equation}
for $Z\sim\mathcal{N}(0,\gamma I)$. Using Stein's lemma (i.e., Gaussian integration by parts) this is the same as
\[dX_t =\frac{\mathbb{E}_Z [Z\cdot f(X_t+\sqrt{1-t}Z)]}{ \sqrt{1-t}\cdot\mathbb{E}_Z [f(X_t+\sqrt{1-t}Z)]}dt+ \sqrt{\gamma} dW_t, \quad X_0=0\, .\]
At $t=1$ the backward heat semigroup/convolution kernel of \eqref{eqn:sde_schrodinger} localizes, but the crucial difference from (overdamped) Langevin dynamics is that it reaches target $\pi$ in finite time, compared to Langevin that reaches target as $t\rightarrow \infty$ in infinite time horizon (but has arbitrary initialization under ergodicity). And without the drift (i.e., the control), one gets Brownian motion which indeed becomes $\mathcal{N}(0,\gamma I)$ at time $t=1$. 

In continuous time, no convexity assumption on $f$ is needed for convergence, thanks to the optimal stochastic control interpretation \cite{Raginsky}. The general problem of arbitrary endpoints with general reference measure will involve forward-backward iterative scheme for reaching a solution, but the particular case under consideration has a convenient analytical form \eqref{eqn:sde_schrodinger}. In the case of Wiener measure as the reference measure, the solution to the Schr\"odinger bridge problem is also intimately connected to the entropy-regularized optimal transport (with quadratic cost) between the two time marginals. 

The following is a sanity check that discretization of the SDE is stable for the particular choice of $f$ as demanded by \cref{lem:random_design}, therefore one could hope to simply implement the inner step \eqref{eqn:gibbs_1} of the Gibbs sampler via e.g., Euler-Maruyama discretization:
\begin{equation}
\label{eqn:em_schrodinger}
X_{k+1} = X_k+h\frac{\frac{1}{S}\sum_{i=1}^S v_i\cdot f(X_k+\sqrt{1-kh}v_i)}{ \sqrt{1-kh}\cdot \frac{1}{S}\sum_{i=1}^S f(X_k+\sqrt{1-kh}v_i)} + \sqrt{\gamma h}  Z_k, \quad X_0=0
\end{equation}
for $v_i\sim\mathcal{N}(0,\gamma I)$ and $Z_k\sim\mathcal{N}(0,I)$ independent. Putting things together gives the following algorithm at iteration $k$.
\begin{algorithm}[H]
\caption{Gibbs Sampler for Random Design Spike-and-Slab}
\label{alg:gibbs_random}
\begin{algorithmic}
\REQUIRE{$\beta_k\in \mathbb{R}^p, z_k\in \{0,1\}^p$}
\FOR{$j=1$ to $p$}
\STATE{Sample $z_{k+1}^j\sim\text{Bern}(Q_j)$ for $Q_j=(1+\frac{1-q}{q}\frac{\tau_1}{\tau_0}\exp(\frac{1}{2}(1/\tau_1^2-1/\tau_0^2)\beta_k[j]^2))^{-1}$ in parallel}
\ENDFOR
\STATE{Draw $\beta_{k+1}$ by running $1/h$ steps of \eqref{eqn:em_schrodinger} for $f(\cdot)$ defined in \cref{lem:schrodinger} where $D(z_{k+1})=\text{Diag}(z_{k+1}\tau_1^{-2}+(1_p-z_{k+1})\tau_0^{-2})$}
\RETURN $\beta_{k+1},z_{k+1}$
\end{algorithmic}
\end{algorithm}
\begin{lemma}
\label{lem:schrodinger}
Between $t\in(0,1)$, for any $n>2$ and $\sigma>0$
\[f(\beta_t)=\frac{\sigma^n}{(\|\beta_t\|^2+\sigma^2)^{n/2}}\exp(\|y\|^2\frac{ \|\beta_t\|^2}{2\sigma^4+2\sigma^2\|\beta_t\|^2}-\frac{1}{2}\beta_t^\top  D(z) \beta_t+\frac{1}{2\gamma}\|\beta_t\|^2)\]
and the drift $b(\beta_t,t):=\nabla_\beta \log \mathbb{E}_{Z\sim\mathcal{N}(0,\gamma I)} [f(\beta_t+\sqrt{1-t}Z)]$ is Lipschitz in $\beta_t$, assuming $\frac{1}{\tau_1^2}I \preceq D(z) \preceq \frac{1}{\tau_0^2} I$ and $\gamma > \tau_0^2$.
\end{lemma}
\begin{proof}
The goal is to show that $\|b(\beta_t^1,t)-b(\beta_t^2,t)\|\leq C \|\beta_t^1-\beta_t^2\|\, \forall \beta_t^1,\beta_t^2$, or equivalently, $\|\nabla_\beta b(\beta_t,t)\|_{op}\leq C$, for any $t\in(0,1)$. We will need the following fact: if $f(\beta_t) > 0$ is $L$-Lipschitz, the convolved quantity $g(\beta_t):=\mathbb{E}_{Z} [f(\beta_t+\sqrt{1-t}Z)]>0$ will be Lipschitz and smooth. To see this, denote the Gaussian density with covariance $\gamma\cdot I$ as $u_\gamma$, since
\[g(\beta_t)=\int f(\beta_t+\sqrt{1-t}y)u_{\gamma}(y) dy=\int f(\beta_t-y)u_{(1-t)\gamma}(y) dy = f*u_{(1-t)\gamma}\]
is a positively-weighted linear combination of shifted $f$, it is clear that it will also be $L$-Lipschitz. Now for the smoothness claim $\|\nabla^2 g(\beta_t)\|_{op}\leq L/\sqrt{(1-t)\gamma}$, we compute since $\|\nabla f\|\leq L$, for any $\|v\|=1$,
\begin{align*}
|v^\top \nabla^2 g(\beta_t)v| &= |v^\top (\nabla f * \nabla u_{(1-t)\gamma}) v|\\
&= |\int (v^\top \nabla f(y))\cdot (\frac{y-\beta}{(1-t)\gamma})^\top v \cdot u_{(1-t)\gamma}(\beta-y)\, dy|\\
&\leq \frac{L}{\sqrt{(1-t)\gamma}} |\int (\frac{y-\beta}{\sqrt{(1-t)\gamma}})^\top v \cdot u_{(1-t)\gamma}(\beta-y)\, dy |\\
&=\frac{L}{\sqrt{(1-t)\gamma}} \mathbb{E}_{z\sim \mathcal{N}(0,1)}[|z|] = \frac{L}{\sqrt{(1-t)\gamma}}\sqrt{2/\pi}\, .
\end{align*}
This in turn implies that $\|\nabla_\beta b(\beta_t,t)\|_{op} \leq \frac{\|\nabla^2 g(\beta_t)\|_{op}}{g(\beta_t)}+\frac{\|\nabla g(\beta_t)\|^2}{g(\beta_t)^2}\leq C$ since $\|\nabla^2 g(\beta)\|_{op}$ and $\|\nabla g(\beta)\|$ is bounded from above and $g(\beta)$ bounded from below. It remains to check that $f$ is Lipschitz to conclude. Write $f(\beta_t)=\frac{\sigma^n}{(\|\beta_t\|^2+\sigma^2)^{n/2}}\alpha$ where $\alpha$ is shorthand for the $\exp(\cdot)$ term, we have 
\begin{align*}
\|\nabla f(\beta_t)\| \leq  \frac{\alpha \sigma^n \|\beta_t\|}{(\|\beta_t\|^2+\sigma^2)^{n/2}} \left(\frac{\|y\|^2}{\sigma^4+\sigma^2 \|\beta_t\|^2}+\frac{\|\beta_t\|^2}{(\sigma^3+\sigma \|\beta_t\|^2)^2}-\frac{n}{\|\beta_t\|^2+\sigma^2}\right)+(\frac{1}{\gamma}-D(z)) \|\beta_t\|
\end{align*}
and it is easy to see that it is always bounded from above on the domain of $\beta$.
%
\end{proof}

The drift in \eqref{eqn:sde_schrodinger} can also be written as a conditional expectation: $\nabla_x \log \mathbb{E}[f(X_1) \vert X_t=x]$ for $(X_t)_{t\in[0,1]}$ distributed as the prior Wiener process $P$ \cite{Raginsky}. In fact the dynamics can be viewed as $X_t = t\beta+B_t$ for $\beta\sim \pi$ and one reaches target at $t=1$ where $B_t$ is the Brownian bridge on $[0,1]$ (therefore $B_0=B_1=0$) -- this is somewhat related to the stochastic localization dynamics \eqref{eqn:snr}, which we turn to in \Cref{sec:SL-meta}.

\begin{remark}
The SDE \eqref{eqn:sde_schrodinger} also shows up in proximal sampler \cite{prox} as part of the backward heat flow interpretation of the RGO oracle (c.f. Lemma 15/ equation (21) therein), albeit with different initialization (we are initializing from the origin, while \cite{prox} initialize from a Gaussian-convolved version of the target). 
\end{remark}

\subsection{Extension: Spike-and-Slab Logistic Regression}
While the preceding results pertain only to linear regression, we sketch its possible applicability to some GLMs via data augmentation technique. In the case of logistic regression for example, $\forall i\in[n]$, $y_i\in \{0,1\}$ with sparse $\beta^*$,
\[y_i | x_i,\beta \sim \text{Bern}\left(\frac{\exp(x_i^\top \beta)}{1+\exp(x_i^\top \beta)}\right)\]
where $x_i$ is the $i$-th row of the matrix $X$. Through the introduction of the auxiliary variable $\omega_i$, one can write the quasi-likelihood as (note the resemblance to linear model after transformation)
\[\mathbb{P}(y_i=1|\omega_i,z,\beta)=\frac{1}{\sqrt{2}}\exp\left((y_i-\frac{1}{2})(x_{i,z}^\top\beta_{z}) -\frac{\omega_i}{2} (x_{i,z}^\top\beta_{z})^2\right) \]
for $\omega_i\sim \text{PG}(1,0)$ the P\'olya-Gamma distribution, which admits efficient sampling algorithm \cite{polya-gamma}. This step relies on the essential integral identity that holds for all $a\in \mathbb{R}$:
\[\frac{(e^\phi)^a}{1+e^\phi}=\frac{1}{2}e^{(a-1/2)\phi}\int_0^\infty \exp(-\omega \phi^2/2) p(\omega) d\omega\, , \]
where $p(\omega)$ is the pdf for PG$(1,0)$.
Assuming a continuous Gaussian spike-and-slab prior on the parameter $\beta$, the Bayesian logistic regression with spike and slab prior has posterior that can be sampled with Gibbs by alternating between
\begin{align}
\pi (\beta\vert z,y,\omega) 
&\propto \exp\left(-\frac{1}{2}(\bar{\beta}^\top \bar{X}^{\top}D(\omega)\bar{X}\bar{\beta}-\bar{\beta}^\top \bar{X}^\top (y-\frac{1}{2})-\bar{\beta}^\top D(\frac{1}{2\tau_1^2})\bar{\beta}\right)\prod_{j=1}^p (\mathcal{N}(\beta_j;0,\tau_0^2))^{1-z_j}\nonumber \\
&\sim \mathcal{N}(\bar{\beta};\Sigma^{-1}\bar{X}^\top (y-\frac{1}{2}),\Sigma^{-1}) \prod_{j=1}^p (\mathcal{N}(\beta_j;0,\tau_0^2))^{1-z_j} \label{eqn: beta_update_logit}
\end{align}
where $\Sigma(z) = \bar{X}^\top D(\omega) \bar{X}+2 D(\frac{z_j}{2\tau_1^2})$ and for each $i\in[n]$ in parallel
\begin{equation}
\pi(\omega_i|\beta,z,y)\sim \text{PG}(1,x_{i,z}^\top \beta_{z})
\end{equation}
and for each $j\in[p]$ sequentially
\begin{align}
\pi(z_j\vert \beta,y,z_{-j},\omega) &\propto \prod_{j=1}^p (q\mathcal{N}(\beta_j;0,\tau_1^2))^{z_j}\cdot((1-q)\mathcal{N}(\beta_j;0,\tau_0^2))^{1-z_j}\mathcal{N}\left(\frac{y-\frac{1}{2}}{\omega};X_{z}\beta_{z}, D(\frac{1}{\omega_i})\right) \nonumber\\
&\propto  (q\mathcal{N}(\beta_j;0,\tau_1^2))^{z_j}\cdot((1-q))^{1-z_j}\times \mathcal{N}(\beta_{z};\Sigma^{-1}\bar{X}^\top (y-1/2),\Sigma^{-1}) \nonumber\\
&\sim \text{Bern}\left(z_j;\frac{q\mathcal{N}(\beta_{z};\Sigma^{-1}\bar{X}^\top (y-1/2),\Sigma^{-1})}{(1-q)\mathcal{N}(\beta_j;0,\tau_0^2)+q\mathcal{N}(\beta_{z};\Sigma^{-1}\bar{X}^\top (y-1/2),\Sigma^{-1})}\right) 
\end{align}
%
%
%
where we used completion of squares at various places. The most expensive step of the update is \eqref{eqn: beta_update_logit}, for which one can re-use similar tricks from \Cref{sec:practical} for speed-up. We leave investigation of the mixing along with statistical property of the posterior for future work.




\section{Stochastic Localization Sampler}
\label{sec:SL-meta}
In this section, we study Stochastic Localization Sampler for \eqref{eqn:quasi_posterior} under similar posterior contraction assumptions with warm start as in \Cref{sec:mixing_gibbs}. This class of samplers essentially takes a denoising perspective -- as we already saw, computationally sampling from the posterior is harder than statistical estimation in some sense (even for identifying the support $z$ as illustrated in \cite{yang2016computational}), but the approach below is not based on MCMC -- therefore \emph{not sensitive} to spectral gap, isoperemetric constant etc. -- and put the two tasks on equal footing under favorable statistical conditions, at least for some spike-and-slab models.

\subsection{Preliminaries: From Denoising to Sampling} 
The idea of stochastic localization came out of the analysis of functional inequalities (i.e., key ingredient behind the solution to the KLS conjecture \cite{yuansi2021KLS}) as a proof technique. The work of \cite{alaoui2022sampling} initiated its algorithmic use for sampling from the Sherrington-Kirkpatrick Gibbs measure with discrete hypercube support $\{\pm 1\}^n$, where approximate message passing (AMP) is used for implementing the mean estimation step, which we explain below (their guarantee holds with probability $1-o_n(1)$ over input $A\sim \text{GOE}(n)$). The crucial insight of this method is that the following two processes have the same law \cite{klartag2021bayesian, alaoui2022sampling} (this is sequential revelation of information)
\begin{equation}
\label{eqn:snr}
\theta_t = t\beta + W_t, \; \beta\sim \pi\; \text{(unknown signal where we know the prior \& have Gaussian observation)}
\end{equation}
which is ideal and un-implementable since we don't know $\beta$, and
\begin{equation}
\label{eqn:sl_sde}
d\theta_t = \left[\int_{\mathbb{R}^p} \beta\cdot p_{t,\theta_t}(\beta) d\beta\right]\, dt + dW_t=\mathbb{E}[\beta|\theta_t=\theta]dt+dW_t,\quad \theta_0 = 0
\end{equation}
for which (notice it only depends on the last time point)
\begin{equation}
\label{eqn:tilted-measure}
p_{t,\theta_t}(\beta) := \frac{1}{Z(t,\theta_t)}\exp(\theta_t^{\top}\beta-\frac{t}{2}\|\beta\|^2) \pi(\beta)
\end{equation}
precisely describes the posterior $\mathbb{P}(\beta \vert (\theta_s)_{0\leq s\leq t} )=\mathbb{P}(\beta \vert \theta_t = \theta)$ for $\beta$ under \eqref{eqn:snr}. Above $Z(t,\theta_t)$ is a normalizing constant. The measure $p_{t,\theta_t}$ localizes to a Dirac measure $\delta_\beta$ for a random $\beta\sim \pi$ as $t\rightarrow \infty$ (this can also be seen from \eqref{eqn:snr} since the signal part scales as $\mathcal{O}(t)$ and the noise part $\mathcal{O}(\sqrt{t})$). We abbreviate $p_{t,\theta_t}$ as $p_t$ below, and let $a_t:= \int \beta\cdot p_t(\beta)d\beta$ that one can think of as a Bayes optimal estimator.

As \cref{lem:alternative} below will reveal, Stochastic Localization is evolving a measure $p_t(\beta)$ driven by $W_t$ that has the martingale property of $p_0=\pi$ and $p_\infty=\delta_\beta$ for $\beta\sim \pi$. The process can be simulated via a SDE \eqref{eqn:sl_sde} which reduces the task of sampling from $\pi$ to estimating the denoising drift $\mathbb{E}[\beta|\theta_t=\theta]$ -- an approximation of this is what we will output at the end after running it for sufficiently long, and we track the (random) evolving measure for its barycenter $a_t$. In some sense at every fixed $t$, the process decomposes $\pi$ into a mixture of random measures, i.e., $\pi = \mathbb{E}_{\theta_t}[\pi(\cdot|\theta_t)]$, and the variance of the component $\pi(\cdot|\theta_t)$ decreases as $t\rightarrow \infty$. A more general version of \eqref{eqn:tilted-measure} can take the form $p_{t,\theta_t}(\beta) = \frac{1}{Z(t,\theta_t)}\exp(\theta_t^{\top}\beta-\frac{\|\beta\|_{G_t}^2}{2}) \pi(\beta)$ for $G_t \succ 0$ but we will not pursue such extension here. 
\begin{remark}
If $\pi(\beta)$ has bounded second moment, $a(\theta_t,t)$ is Lipschitz in $\theta_t$, since
\begin{align*}
\|\nabla_{\theta_t} a(\theta_t,t)\|_{op}=\|\mathbb{E}[\beta \beta^\top]-\mathbb{E}[\beta]\mathbb{E}[\beta]^\top \|_{op} \leq \|\mathbb{E}[\beta \beta^\top]\|_{op} \leq \mathbb{E}[\|\beta\|^2]
\end{align*}
will be bounded, where above the expectation is taken with respect to $p_{t,\theta_t}(\beta)$, which means the SDE \eqref{eqn:sl_sde} has a unique strong solution. 
\end{remark}

The lemma below gives quantitative convergence rate for \eqref{eqn:sl_sde} in continuous time.
\begin{lemma}[Continuous time SDE convergence]
\label{lem:continuous_time}
We have after $t=1/\epsilon^2$, 
\[W_2(\pi,\text{Law}(a_t))=W_2(\mathbb{E}[p_t],\text{Law}(a_t))\leq \sqrt{p}\epsilon\,.\]
\end{lemma}
\begin{proof}
Based on covariance decay we have $\mathbb{E}[\text{cov}(\mu_t)]\preceq \frac{1}{t}I$ for all $t>0$ \cite{alaoui2022sampling}, which reflects the fact that the measure localizes, therefore
\[\mathbb{E}[W_2^2(p_t,\delta_{a_t})]\leq \mathbb{E}[\mathbb{E}_{p_t}[\|x-a_t\|^2]]\leq \frac{p}{t}\]
by the coupling definition of $W_2$ distance and taking trace on both sides. Now since $W_2^2$ is convex (this can be seen from the dual formulation which is $\sup$ over a set of linear functions), we can push expectation inside using Jensen's inequality and conclude $W_2^2(\mathbb{E}[p_t],\text{Law}(a_t))\leq \frac{p}{t}$. Recall $\mathbb{E}[\mathbb{E}_{x\sim p_t}[x]] = \mathbb{E}_{x\sim p_0}[x]=\mathbb{E}_{x\sim \pi}[x]$ from the martingale property, hence $\text{Law}(a_t)\rightarrow p_0 = \pi$ as $t\rightarrow \infty$.
\end{proof} 
This rate is slower than other SDE-based algorithms, which have exponential convergence in continuous time under strong convexity, but is nevertheless quite minimal in terms of the assumptions made. 
%
%
In fact there's a dynamics one can write for the barycenter as well, if one can compute the covariance of $p_{t}(\beta)$. This is a drift-free, diffusion-only SDE with multiplicative noise, which could make discretization challenging. In this regard \eqref{eqn:sl_sde} relies on the mean and \eqref{eqn:barycenter} relies on the covariance of the tilted measure for implementation. 
\begin{lemma}[Alternative Descriptions]
\label{lem:alternative}
We have the following barycenter representation:
\begin{equation}
\label{eqn:barycenter}
da_t = A_t \, dW_t= \left[\int_{\mathbb{R}^p} (\beta-a_t)(\beta-a_t)^{\top}p_{t,\theta_t}(\beta) \, d\beta\right] \, dW_t
\end{equation}
where $a_\infty\sim \pi$. And density-valued SDE representation:
\begin{equation}
\label{eqn:density_martingale}
dp_t(\beta) = p_t(\beta)\langle \beta-a_t,dW_t\rangle\, . 
\end{equation}
\end{lemma}
\begin{proof}
These are known in the context of stochastic localization so we simply refer the reader to \cite{eldan_note} for the proof. As an immediate consequence of the martingale property, which is evident from \eqref{eqn:density_martingale}, we have $\mathbb{E}[\int f(\beta) p_{t,\theta_t}(\beta)d \beta]=\int f(\beta) \pi(\beta) d\beta$ remains constant for all $t\geq 0$ for any continuous function $f$. Therefore if $\pi$ has bounded mean/second moment, $(\beta_t)_t$ will have similarly bounded mean/second moment in expectation throughout the localization process.
\end{proof}

The density-valued SDE could potentially be used for an ensemble / interacting particle system implementation on a fixed grid with $\delta_{\beta_1},\delta_{\beta_2}, \dots$, but it will likely require a fine grid for the localization on the continuous domain that we consider (i.e., exponential in dimension). We will not explore it here but nevertheless establish its validity: if we start with a probability distribution $p_0=\sum_i p_0(\beta_i)=1$, the process will remain a probability measure over the discrete set since $a_t = \sum_i \beta_i p_t(\beta_i)$, for $\beta_i\in \mathbb{R}^p\, \forall i$,
\[\frac{d \sum_i p_t(\beta_i)}{dt} = \sum_i p_t(\beta_i)\langle \beta_i-a_t,dW_t\rangle =0 \Rightarrow \sum_i p_t(\beta_i)=1 \quad \forall t>0\, .\]
The time-discretized algorithm for sampling from $\pi$ using \eqref{eqn:sl_sde} is given below. We note that the algorithm is in some sense gradient-free.
\begin{algorithm}[H]
\caption{Stochastic Localization Sampler for $\pi$}
\label{alg:SL}
\begin{algorithmic}
\REQUIRE{Blackbox $\mathcal{T}(\theta_t)$ that can (approximately) compute $\int_{\mathbb{R}^p} \beta\cdot p_{t,\theta_t}(\beta) d\beta$ from \eqref{eqn:tilted-measure}}
\STATE{Initialize $\beta_1 = 0$}
\FOR{$k=1$ to $K$}
\STATE{$\beta_{k+1}=\beta_k+h \mathcal{T}(\beta_k)+\sqrt{h} z_k$ for $z_k\sim\mathcal{N}(0,I)$ independent}
\ENDFOR
\RETURN{$\mathcal{T}(\beta_{K+1})$}
\end{algorithmic}
\end{algorithm}

\subsection{Warm-up: Orthogonal Design} 
In the case of $X=I$ (sequence model), since we start with a product measure, we end up with another product measure that decouples across coordinates, which reduces the complexity significantly. With the point-mass spike and slab prior, the marginal posterior distribution of each coordinate is a mixture of (data-dependent, weighted) Dirac measure at zero and a continuous convolved density, with the weights signaling if the parameter has a higher chance coming from the spike or the slab part given the data and a fixed $q$: 
\begin{align}
\pi(\beta_j|y_j,q)&=\mathbb{P}(z_j=1|y_j,q)\pi(\beta_j|y_j,z_j=1)+\mathbb{P}(z_j=0|y_j,q)\pi(\beta_j|y_j,z_j=0) \nonumber\\
&= \frac{(1-q)\phi_{\sigma}(y_j)}{(1-q)\phi_{\sigma}(y_j)+q h(y_j)}\delta_0(\beta_j)+\frac{q h(y_j)}{(1-q)\phi_{\sigma}(y_j)+q h(y_j)} \frac{\phi_{\sigma}(y_j-\beta_j) g_{\tau_1}(\beta_j)}{\int \phi_{\sigma}(y_j-\beta_j) g_{\tau_1}(\beta_j) d\beta_j} \label{eqn:marginal}
\end{align}
where $\phi_{\sigma}(y_j-\beta_j)\propto e^{-\frac{1}{2\sigma^2}(y_j-\beta_j)^2}$ is the likelihood, $h(y_j):=\int \phi_{\sigma}(y_j-\beta_j) g_{\tau_1}(\beta_j) d\beta_j$ the convolution and $g_{\tau_1}(\cdot)$ the slab prior. For the choice of $q\geq 1/p$, a known fact is that the posterior median behaves similarly as a coordinate-wise hard thresholding estimator with threshold $\sigma\sqrt{2\log(p)}$, i.e., the max of $p$ independent Gaussians with variance $\sigma^2$, which capture the level below which there is no expected signal. It has been recognized since the 90s that shrinkage estimator can be tuned to attain minimax rates over a wide range of sparsity classes \cite{donoho_threshold}. The empirical Bayes choice of $q$ can be performed by maximizing the log-marginal $q \vert y$ as $\arg\max_q \sum_{j=1}^n \log((1-q)\phi_{\sigma}(y_j)+q h(y_j))$ but we will not pursue such an extension here.
 
 We remark that sequence model is known to be polynomial-time computable -- even with a hyper-prior on $q$ that renders the coordinates dependent, existing exact method scales as $\mathcal{O}(n^3)$ using polynomial multiplication \cite{castillo_sequence} for calculating various posterior point estimators. In what follows in this section we assume the data matrix satisfies $X^{\top}X = I_p, n=p$, i.e., orthogonal, since some salient features of the dynamics can be more easily seen in this simpler case.
Under point-mass spike, by definition, given $t,\theta_t,y$ the mean of the tilted measure is given by  
\begin{align*}
a_{t}(\theta_t,t) &= \int_{\mathbb{R}^p} \beta \cdot p_{t,\theta_t}(\beta) \, d\beta\\
&= \frac{1}{Z} \int_{\mathbb{R}^p} \beta \cdot \sum_{z\in \{0,1\}^p} e^{\theta_t^{\top}\beta-\frac{t\|\beta\|^2}{2}-\frac{1}{2\sigma^2}\|y-X_z\beta_z\|_2^2} \prod_{j=1}^p (q\mathcal{N}(\beta_j;0,\tau_1^2))^{z_j}\cdot((1-q)\delta_0(\beta_j))^{1-z_j} \, d\beta\, .  
\end{align*}
 Without loss of generality we look at the first coordinate. 
 %
%
 Let $x_i$ denote the $i$-th column of the matrix $X$, for point-mass spike whether we assume quasi-likelihood or exact likelihood doesn't affect the calculation in this case. Recall $a_{t,1}$ can be viewed as a denoiser for $\beta_1^*$ and $a_{t,1}\rightarrow \delta_{\beta_1^*}$ for some $\beta_1^*\sim \pi$ as $t\rightarrow \infty$, which we output.
\begin{align*}
a_{t,1}(\theta_t,t) &= \frac{\int_{\mathbb{R}}\beta_1\cdot \exp\left((\theta_{t,1}+\frac{1}{\sigma^2} y^{\top}x_1)\beta_1-(\frac{t}{2}+\frac{1}{2\sigma^2})\beta_1^2 \right)[q\frac{1}{\sqrt{2\pi}\tau_1}e^{-\frac{\beta_1^2}{2\tau_1^2}}+(1-q)\delta_0(\beta_1)]\, d\beta_1}{\int_{\mathbb{R}} \exp\left((\theta_{t,1}+\frac{1}{\sigma^2} y^{\top}x_1)\beta_1-(\frac{t}{2}+\frac{1}{2\sigma^2})\beta_1^2 \right)[q\frac{1}{\sqrt{2\pi}\tau_1}e^{-\frac{\beta_1^2}{2\tau_1^2}}+(1-q)\delta_0(\beta_1)]\, d\beta_1 }\\
&= \frac{q\frac{1}{\sqrt{2\pi}\tau_1}\int_{\mathbb{R}}\beta_1\cdot \exp\left((\theta_{t,1}+\frac{1}{\sigma^2} y^{\top}x_1)\beta_1-(\frac{t}{2}+\frac{1}{2\sigma^2}+\frac{1}{2\tau_1^2})\beta_1^2 \right)\, d\beta_1}{q\frac{1}{\sqrt{2\pi}\tau_1}\int_{\mathbb{R}} \exp\left((\theta_{t,1}+\frac{1}{\sigma^2} y^{\top}x_1)\beta_1-(\frac{t}{2}+\frac{1}{2\sigma^2}+\frac{1}{2\tau_1^2})\beta_1^2 \right)\, d\beta_1 + (1-q)}\\
&= \frac{\theta_{t,1}+\frac{1}{\sigma^2}y^{\top}x_1}{(t+\frac{1}{\sigma^2}+\frac{1}{\tau_1^2})+\frac{1-q}{q}(t+\frac{1}{\sigma^2}+\frac{1}{\tau_1^2})^{3/2}\tau_1\exp(-\frac{(\theta_{t,1}+\frac{1}{\sigma^2}y^{\top}x_1)^2}{2(t+\frac{1}{\sigma^2}+\frac{1}{\tau_1^2})})}
\end{align*}
where we used $\int_{-\infty}^\infty x\exp(-ax^2+bx)\, dx = \frac{\sqrt{\pi}b}{2a^{3/2}}\exp(b^2/4a)$ and $\int_{-\infty}^\infty \exp(-ax^2+bx)\, dx = \sqrt{\frac{\pi}{a}}\exp(b^2/4a)$ for $a>0$. The effect of spike is to introduce shrinkage -- in particular if we look at the denominator, it only becomes prominent when (for $q\geq 1/p$)
\[|\theta_{t,1}+\frac{1}{\sigma^2}y^\top x_1|\leq \sqrt{(t+1/\sigma^2+1/\tau_1^2)\log(\tau_1^2 p^2(t+1/\sigma^2+1/\tau_1^2))}\, ,\]
and in the case $X=I$, $y^\top x_1 = y_1$. For small $t$, this gives the threshold for $|y_1| \lesssim \sigma \sqrt{2\log(p\tau_1/\sigma)}$; and for large $t$, this becomes $|\theta_{t,1}|\lesssim \sqrt{2t\log(\tau_1 p \sqrt{t})}$. For the sampling dynamics
\[d\beta_{t,1} = a_{t,1}(\beta_t,t)dt+dW_t \, ,\]  
we see that initially if $|y_1|$ is above the threshold, it behaves almost like a linear SDE with time-dependent drift $\frac{\beta_{t,1}+\frac{1}{\sigma^2}y^{\top}x_1}{t+\frac{1}{\sigma^2}+\frac{1}{\tau_1^2}}$ that can be integrated exactly and $\beta_{t,1}$ scales as $\sim t$; otherwise the Brownian motion part will take over and $\beta_{t,1}$ roughly scales as $\sim \sqrt{t}$. As $t\rightarrow \infty$, with all else holding constant (i.e., for any finite sample size $n$), the drift 
\[a_{t,1}(\beta_t,t)\approx \frac{\beta_{t,1}+\frac{1}{\sigma^2}y^\top x_1}{t+\frac{1}{\sigma^2}+\frac{1}{\tau_1^2}}\approx \frac{t\beta_1^*+W_t}{t}\rightarrow \beta_1^*\sim \pi_1 \]
if $\beta_{t,1}\gtrsim \sqrt{t}$, signaling it will converge to the slab part of the posterior \eqref{eqn:marginal}; otherwise if $\beta_{t,1}\lesssim \sqrt{t}$,
\[a_{t,1}(\beta_t,t)\approx \frac{\beta_{t,1}+\frac{1}{\sigma^2}y^\top x_1}{\frac{1-q}{q}(t+\frac{1}{\sigma^2}+\frac{1}{\tau_1^2})^{3/2}\tau_1}\approx \frac{\beta_{t,1}}{t^{3/2}}\rightarrow 0\,.\]
%
%
%
On the other hand, with Gaussian spike and sparsified likelihood \eqref{eqn:quasi_posterior}, for $\tau_1 \gg \tau_0$,
\begin{align*}
a_{t,1}(&\theta_t,t)=\frac{\int\beta_1\cdot e^{(\theta_{t,1}+\frac{1}{\sigma^2} y^{\top}x_1)\beta_1-(\frac{t}{2}+\frac{1}{2\sigma^2})\beta_1^2 }q\frac{1}{\sqrt{2\pi}\tau_1}e^{-\frac{\beta_1^2}{2\tau_1^2}} d\beta_1+\int\beta_1\cdot e^{\theta_{t,1}\beta_1-\frac{t}{2}\beta_1^2 }  (1-q)\frac{1}{\sqrt{2\pi}\tau_0}e^{-\frac{\beta_1^2}{2\tau_0^2}} d\beta_1}{\int e^{(\theta_{t,1}+\frac{1}{\sigma^2} y^{\top}x_1)\beta_1-(\frac{t}{2}+\frac{1}{2\sigma^2})\beta_1^2 }q\frac{1}{\sqrt{2\pi}\tau_1}e^{-\frac{\beta_1^2}{2\tau_1^2}}  + e^{\theta_{t,1}\beta_1-\frac{t}{2}\beta_1^2 }(1-q)\frac{1}{\sqrt{2\pi}\tau_0}e^{-\frac{\beta_1^2}{2\tau_0^2}} d\beta_1}\\
&= \frac{\frac{q}{\tau_1}\frac{\theta_{t,1}+1/\sigma^2 y^\top x_1}{(t+1/\sigma^2+1/\tau_1^2)^{3/2}}\exp(\frac{(\theta_{t,1}+1/\sigma^2y^\top x_1)^2}{2t+2/\sigma^2+2/\tau_1^2})+\frac{1-q}{\tau_0}\frac{\theta_{t,1}}{(t+1/\tau_0^2)^{3/2}}\exp(\frac{\theta_{t,1}^2}{2t+2/\tau_0^2})}{\frac{q}{\tau_1}\frac{1}{(t+1/\sigma^2+1/\tau_1^2)^{1/2}}\exp(\frac{(\theta_{t,1}+1/\sigma^2y^\top x_1)^2}{2t+2/\sigma^2+2/\tau_1^2})+\frac{1-q}{\tau_0}\frac{1}{(t+1/\tau_0^2)^{1/2}}\exp(\frac{\theta_{t,1}^2}{2t+2/\tau_0^2})}\\
&= \frac{\theta_{t,1}+1/\sigma^2 y^\top x_1}{(t+1/\sigma^2+1/\tau_1^2)+\frac{1-q}{q}\frac{\tau_1}{\tau_0}\frac{(t+1/\sigma^2+1/\tau_1^2)^{3/2}}{\sqrt{t+1/\tau_0^2}}\exp(\frac{\theta_{t,1}^2}{2t+2/\tau_0^2}-\frac{(\theta_{t,1}+1/\sigma^2 y^\top x_1)^2}{2t+2/\sigma^2+2/\tau_1^2})}\\
&\quad +\frac{\theta_{t,1}}{(t+1/\tau_0^2)+\frac{q}{1-q}\frac{\tau_0}{\tau_1}\frac{(t+1/\tau_0^2)^{3/2}}{\sqrt{t+1/\sigma^2+1/\tau_1^2}}\exp(\frac{(\theta_{t,1}+1/\sigma^2 y^\top x_1)^2}{2t+2/\sigma^2+2/\tau_1^2}-\frac{\theta_{t,1}^2}{2t+2/\tau_0^2})}\, .
\end{align*}
Therefore as $t\rightarrow \infty$, with all else holding constant (i.e., for any finite sample size $n$), one of the above two terms will go to $\frac{\theta_{t,1}}{\exp(t)}=\frac{t\beta_1^*+W_t}{\exp(t)}\rightarrow 0$ and the other go to $\frac{\theta_{t,1}}{t}=\frac{t\beta_1^*+W_t}{t}\rightarrow \beta_1^*\sim \pi_1$, depending on whether
\[\frac{(\theta_{t,1}+1/\sigma^2 y^\top x_1)^2}{t+1/\sigma^2+1/\tau_1^2}\lessgtr \frac{\theta_{t,1}^2}{t+1/\tau_0^2}\, ,\]
if $\theta_{t,1} \gtrsim \sqrt{t}$, which is the only possibility since the posterior $\pi$ puts zero mass at $0$ exactly (with the first $\delta_0(\beta_j)$ term from \eqref{eqn:marginal} replaced by another convolved density $\phi_{\sigma}(y_j-\beta_j) g_{\tau_0}(\beta_j)$). 
%
Consequently, continuous spike-and-slab priors yield non-sparse posterior point estimators that require thresholding for variable selection, and the alternative of selection based on $\mathbb{P}(z|y)$ can be expensive generally.
 


For some intuition on the time discretization of the SDE, take the point-mass spike-and-slab for example, since $\pi$ is sub-Gaussian (therefore Novikov's condition holds with a very similar argument as below), using Girsanov's theorem, and consider the two SDEs:
\begin{align*}
d\beta_t &= a(\beta_t,t)dt+dW_t, \quad \text{same as \eqref{eqn:sl_sde}}\\
d\hat{\beta}_t &= a(\hat{\beta}_{kh},kh) dt+dW_t, \quad \text{for } t\in[kh,(k+1)h] \text{ an interpolation of discrete update \eqref{eqn:discrete_update}}
\end{align*}
where $(\beta_t)_t \sim Q, (\hat{\beta}_t)_t \sim P$ are two path measures, and one can obtain with the data processing inequality, 
\begin{align*}
\text{KL}(\pi_{Kh} || \mu_{Kh})&\leq \text{KL}(Q_{Kh} || P_{Kh})\\
&\lesssim \sum_{k=1}^{K} \int_{kh}^{(k+1)h} \mathbb{E}_Q[\|a(\beta_t,t)-a(\beta_{kh},kh)\|^2] dt\\
&\lesssim L(\sigma,h,\tau_0,\tau_1,y) \sum_{k=1}^{K} \int_{kh}^{(k+1)h} \mathbb{E}_Q[\|\beta_t-\beta_{kh}\|^2] dt\\
&\lesssim L(\sigma,h,\tau_0,\tau_1,y) \sum_{k=1}^{K} \int_{kh}^{(k+1)h} [(t-kh)^2\mathbb{E}_Q[\|a(\beta_t,t)\|^2] + 2d (t-kh)] dt
\end{align*}
which means if $h$ and $K$ are sufficiently small, since using Jensen's inequality, the drift  
\[\mathbb{E}_Q[\|a(\beta_t,t)\|^2]= \mathbb{E}_Q[\|\int \beta 
 p_{t,\theta_t}(\beta) d\beta\|^2]\leq \mathbb{E}_Q[\int \|\beta\|^2 p_{t,\theta_t}(\beta) d\beta] = \int \|\beta\|^2 \pi(\beta) d\beta < \infty\]
along the dynamics as shown in \cref{lem:alternative}, the two processes will be close to each other in law. Above $L$ is a constant depending on $\sigma, h, \tau_0,\tau_1,y$ since each coordinate $a_j(\beta_t,t)$ can be written as for some $c(0),c(1) > 0$, 
\[\min\{v(0),v(1)\}\leq \frac{v(0)c(0)+v(1)c(1)}{c(0)+c(1)}\leq \max\{v(0),v(1)\}\]
where $v(1) = (1/\sigma^2+1/\tau_1^2+t)^{-1}(1/\sigma^2 y^\top x_i+\beta_{t,i})$ and $v(0)=(1/\tau_0^2+t)^{-1}\beta_{t,i}$, and similarly for $a_j(\beta_{kh},kh)$ therefore $\|a(\beta_t,t)-a(\beta_{kh},kh)\|$ can be bounded by the claimed quantities. Notice that above we didn't use any approximations for $a(\cdot)$ -- since the computation scales linearly with $p$ instead of exponentially in this case, we didn't rely on probabilistic arguments / large-scale behavior on the model for showing convergence of the time-discretized SDE \eqref{eqn:sl_sde} for sampling from $\pi$ (of course, for $\pi$ to behave well statistically however, $\tau_1,\tau_0, q$ will have to be chosen carefully as we will see in \Cref{sec:stat_normal_mean}).

\subsection{Spike-and-Slab Linear Regression: Mean Computation}
Recall from \cref{prop:sparsify} the posterior marginal over $\beta$ in this case is a discrete mixture of log-concave densities: 
\begin{equation}
\label{eqn:posterior1}
\pi(\beta\vert y)\propto \sum_{z\in \{0,1\}^p} q^{\|z\|_0}(1-q)^{p-\|z\|_0}\times \frac{e^{-\frac{1}{2}\beta^\top D_z^{-1}\beta}}{\sqrt{\det(2\pi D_z)}} \times \frac{e^{-\frac{1}{2\sigma^2}\|y-X_{z}\beta_{z}\|^2}}{(2\pi\sigma^2)^{n/2}}
\end{equation} 
where $D_z$ is diagonal with $\tau_1^2$ if $z=1$ and $\tau_0^2$ otherwise ($\tau_1 \gg \tau_0$), and we will adopt the same assumption as in \Cref{sec:mixing_gibbs} that the data/posterior belong to $\mathcal{E}_s$ implying posterior concentration with the initial number of false positives $t$ bounded (the design matrix $X$ is again assumed deterministic satisfying the same ``restricted isometry" conditions). We note that the posterior \eqref{eqn:posterior1} is non-convex / non-smooth so $\arg\max$ (MAP) estimator is also hard to obtain from optimization, but integration/sampling can be somewhat easier under favorable statistical assumptions.

\begin{lemma}
\label{lem:mean}
For the sparsified likelihood with continuous priors \eqref{eqn:quasi_posterior}, we have given fixed $t,\theta_t,y, q\in (0,1)$ the approximate drift 
\begin{align}
\label{eqn:approximate_a}
\hat{a}(\theta_t,t) = \frac{\sum_{z: z \in \mathcal{S}} v(z)\cdot c(z)}{\sum_{z: z \in \mathcal{S}} c(z)} \, ,
\end{align} 
where $v(z),c(z)$ are defined in \eqref{eqn:linear-system}-\eqref{eqn:c_z}, and $\mathcal{S}$ is the warm start set with $z^*\subset z, \|z\|_0\leq k+t$. Recall from \cref{lem:warm-start} that a warm-start with number of false positives $t\asymp k$ can generally be expected under \cref{assume:scaling}, therefore $\sum_{i=0}^k {t+k\choose i}\leq (e(t+k)/k)^k\asymp ((t+k)/k)^k\asymp c^t$ number of sub-models are evaluated at each time step. Additionally, under the statistical assumptions for \Cref{prop:stat}, $\frac{1}{\sqrt{p}}\|\hat{a}(\theta_t,t)- a(\theta_t,t)\|$ converges to zero in probability as $n\rightarrow \infty$, as the rest of $z^* \not\subset z$ contributes vanishingly small to the posterior.  
\end{lemma}
\begin{proof}
By definition the tilted mean (as a function of the random measure $\pi$) takes the form
\begin{align*}
&a(\theta_t,t)= \sum_{z\in\{0,1\}^p}\int_{\mathbb{R}^p} \beta \cdot p_{t,\theta_t}(\beta,z) \, d\beta=
\int_{\mathbb{R}^p} \beta\cdot \exp(\theta_t^{\top}\beta-\frac{t\|\beta\|^2}{2}) \pi(\beta)\, d\beta\,\\
& = \frac{\sum_{z \in \{0,1\}^p} q^{\|z\|_0}(1-q)^{p-\|z\|_0}\left[\int_{\mathbb{R}^p} \beta \cdot e^{\theta_t^{\top}\beta-\frac{t\|\beta\|^2}{2}-\frac{1}{2\sigma^2}\|y-X_{z}\beta_{z}\|_2^2-\frac{1}{2}\beta^\top D_z^{-1}\beta} d\beta\right]\frac{1}{\sqrt{\det(2\pi D_z)}}}{\sum_{z\in\{0,1\}^p}  q^{\|z\|_0}(1-q)^{p-\|z\|_0}\times \int_{\mathbb{R}^p} e^{\theta_t^{\top}\beta-\frac{t\|\beta\|^2}{2}}\times \frac{e^{-\frac{1}{2}\beta^\top D_z^{-1}\beta}}{\sqrt{\det(2\pi D_z)}} \times e^{-\frac{1}{2\sigma^2}\|y-X_{z}\beta_{z}\|^2} d\beta}\\
&= \frac{\sum_{z \in \{0,1\}^p} v(z)\cdot c(z)}{\sum_{z\in\{0,1\}^p}c(z)}
\end{align*}
for vector $v(z) \in \mathbb{R}^p$,
\begin{equation}
\label{eqn:linear-system}
v(z)_j = \begin{cases}
[(\frac{1}{\sigma^2}X_{z}^\top X_{z}+\frac{1}{\tau_1^2}I+tI)^{-1}(\frac{1}{\sigma^2}X_{z}^\top y+\theta_{t,z})]_j \quad \text{if $j$ is active}\\
[(\frac{1}{\tau_0^2}I+tI)^{-1}\theta_{t,1-z}]_j \quad \text{otherwise}\, ,
\end{cases}
\end{equation}
furthermore the scalar 
\begin{align}
c(z) &= \frac{\exp\left(\frac{1}{2}(\frac{1}{\sigma^2} y^\top X_{z}+\theta_{t,z}^\top) (\frac{1}{\sigma^2}X_{z}^\top X_{z}+\frac{1}{\tau_1^2}I+tI)^{-1}(\frac{1}{\sigma^2}X_{z}^\top y+\theta_{t,z})\right)}{\sqrt{\det(\frac{1}{\sigma^2}X_{z}^\top X_{z}+\frac{1}{\tau_1^2}I+tI)}} \label{eqn:c_z}\\
&\times \frac{\exp(\frac{1}{2}\theta_{t,1-z}^\top(\frac{1}{\tau_0^2}I+tI)^{-1}\theta_{t,1-z})}{\sqrt{\det(\frac{1}{\tau_0^2}I+tI)}}\times (\frac{q\tau_0}{(1-q)\tau_1})^{\|z\|_0} \nonumber
\end{align}
where we used Gaussian integral and completion of squares. 
The approximate posterior mean which acts as the drift of the SDE \eqref{eqn:sl_sde} is given by for the warm start set $\mathcal{S}:=\{z: z^*\subset z, \|z\|_0\leq \|z^*\|_0+t\}$ with at most $k+t \ll p$ active coordinates,
\begin{equation*}
\hat{a}(\theta_t,t) = \frac{\sum_{z: z \in \mathcal{S}} v(z)\cdot c(z)}{\sum_{z: z \in \mathcal{S}} c(z)}
\end{equation*}
where computing \eqref{eqn:linear-system} involves solving linear systems of size $\|z\|_0\times \|z\|_0$ with both changing left and right hand sides $t$ and $\theta_t$. Asymptotically as $t\rightarrow \infty$, the drift becomes $z$-independent and approaches $\theta_t/t=\beta+W_t/t \rightarrow \beta$ for some random $\beta\sim \pi$, which we output. From the denoising perspective, the task gets easier as $t\rightarrow \infty$ since the signal-to-noise ratio grows like $t/\sqrt{t}=\sqrt{t}$.

Using \eqref{eqn:high_prob_selection} as a consequence of \Cref{prop:stat}, we have $\forall \epsilon > 0$, recall $p_n\rightarrow \infty$ as $n\rightarrow \infty$ such that $p_n=e^{o(n)}$, since $z^*\in \mathcal{S}$ by definition, 
\begin{align*}
&\lim_{n\rightarrow \infty}\mathbb{P}(\frac{1}{\sqrt{p}}\|\hat{a}(\theta_t,t)-a(\theta_t,t)\| \geq \epsilon) \\
&\leq \lim_{n\rightarrow \infty}\mathbb{P}(\frac{1}{\sqrt{p}}\|\hat{a}(\theta_t,t)-a(\theta_t,t)\| \geq \epsilon | \pi(z^*|y)\geq 1-1/p)+\lim_{n\rightarrow \infty}\mathbb{P}(\pi(z^*|y) \leq 1-1/p)\\
&\leq 0+\lim_{n\rightarrow \infty}  \frac{1}{p} = 0
\end{align*}
therefore $\text{p-}\lim_{n\rightarrow \infty} \hat{a}(\theta_t,t)_j = \text{p-}\lim_{n\rightarrow \infty} a(\theta_t,t)_j$ yields the convergence in probability claim.
\end{proof}

We can use pre-computation scheme and cache a factorization of $X_{z}^\top X_{z}$ (generally expected to be full rank since $k\lesssim n$) to speed up the subsequent calculation. Since the sub-models under consideration share common features, one should also use Sherman-Morrison for low-rank updates whenever possible. 

\begin{remark}
If the integral is hard to compute analytically, one might hope to use Laplace approximation. It may also be possible to use mode instead of mean if the posterior consists of mixture of log-concave distributions (they can be shown to be not far apart due to measure concentration for log-concave densities), in the case of more general slab distributions. 
\end{remark}

\subsection{Spike-and-Slab Linear Regression: SDE Implementation}
Recall we discretize as  
\begin{equation}
\label{eqn:discrete_update}
\beta_{k+1} = \beta_k+h\cdot \hat{a}(\beta_k,kh)+\sqrt{h}\cdot z_k, \quad z_k\sim\mathcal{N}(0,I) \; \text{independent}
\end{equation}
and output $\hat{a}(\beta_k,kh)$ for sufficiently large $k$. In line with \Cref{sec:statistics} we consider a sequence of problems with growing $n,p_n,k_n \rightarrow \infty$, so the posterior is implicitly indexed by $n$, and the probabilities are conditional on $X$. Here $p_n/n\sim e^{o(n)}/n, k_n/n\sim \log(p_n)/n\sim o(n)/n$ serve as proxies for statistical difficulty of the problem, which cannot grow too fast. This is a more meaningful limit than the classical fixed $p$, large $n$ setup. We are interested in the regime where one has variable-selection consistency in the sense $\mathbb{E}[\pi(z^*|y)] \geq 1-\frac{1}{p^2}$, which is established in \cref{prop:stat} under appropriate parameter choices (the allowed scaling of $p,k$ will depend on $X$ for such a guarantee to hold). We study the convergence rate of the Stochastic Localization sampler in this setting -- in fact a guarantee of both computational \& statistical nature along the lines of $\mathbb{E}_{\beta^*}(\mathbb{P}_n(\|\hat{a}(\beta_K)-\beta^*\| \lesssim M \vert y^n)) \geq 1-o_n(1)$ should also be within-reach for the output of the algorithm under such posterior contraction.  

The helper lemma below on the exact drift is crucial for the stable discretization of the SDE, where we borrow parts from \cite[Lemma 4.9]{alaoui2022sampling}.  
\begin{lemma}[Lipschitz-type property of $a(\cdot)$]
\label{lem:regular}
For some constant $C$ depending on $t$, the following regularity condition on $\beta(t) \mapsto  a(\beta(t),t)$ holds: for any $h \leq t\leq T$ and $\beta_k,\beta_t\in \mathbb{R}^p$, with probability $1-o_n(1)$ over the data $y^n$,
\[\|a(\beta_k,t)-a(\beta_t,t)\| \leq C(t)  \|\beta_k-\beta_t\|+o_n(1)\, .\]
Moreover with $(k+1)h\leq T$, for the continuous process \eqref{eqn:continuous_couple} on $\bar{\beta}(t)$, and sufficiently small $h$ such that $h < \lambda_{\min} (X_{z^*}^\top X_{z^*})/\sigma^2$,
\[\sup_{t\in [kh,(k+1)h]} \frac{1}{\sqrt{p}}\|a(\bar{\beta}(t),t)-a(\bar{\beta}(kh),kh)\| = O_p(\sqrt{h}) \, .\]  
Above both are stated under the assumptions for \cref{prop:stat}.
\end{lemma}
\begin{proof}
Since $\tau_1\rightarrow \infty, \tau_0\rightarrow 0$ as $n\rightarrow \infty$, which ensures $\pi(z=z^*|y)\rightarrow 1$ as $n\rightarrow \infty$ from \cref{prop:stat}, recall we have for any given $\beta_t,t$,
\begin{align*}
&v(z)_j =\\
&\begin{cases}
 [(\frac{1}{\sigma^2}X_{z}^\top X_{z}+\frac{1}{\tau_1^2}I+tI)^{-1}(\frac{1}{\sigma^2}X_{z}^\top y+\beta_{t,z})]_j \rightarrow [(\frac{1}{\sigma^2}X_{z}^\top X_{z}+tI)^{-1}(\frac{1}{\sigma^2}X_{z}^\top y+\beta_{t,z})]_j\\
[(\frac{1}{\tau_0^2}I+tI)^{-1}\beta_{t,{1-z}}]_j \rightarrow 0 
\end{cases}
\end{align*}
as $n\rightarrow \infty$ and for some $c(z)\geq 0$,
\[\min_z v(z)_j \leq a(\beta_t,t)_j=\frac{\sum_{z \in \{0,1\}^p} v(z)_j\cdot c(z)}{\sum_{z\in\{0,1\}^p}c(z)} \leq \max_z v(z)_j\, .\]
Now for any $t \geq h$, with probability $1-o_n(1)$, since $X_{z^*}^\top X_{z^*}\succ 0$,
\[\|a(\beta_k,t)-a(\beta_t,t)\| \leq \|(\frac{1}{\sigma^2}X_{z^*}^\top X_{z^*}+tI)^{-1}\|_{op}\|\beta_{t}-\beta_k\|+o_n(1) \leq \frac{1}{t}\|\beta_k-\beta_t\|+o_n(1)\, .\]
For the second part, using that for two linear systems $Lu=r$ and $\hat{L}\hat{u}=\hat{r}$ where $\|L^{-1}(\hat{L}-L)\|<1$, the perturbed solution obeys
\[\|\hat{u}-u\|\leq \frac{\|L^{-1}\|}{1-\|L^{-1}(\hat{L}-L)\|}(\|(\hat{L}-L)u+\hat{r}-r\|)\,;\]
with probability taken over both the stochastic process $(\bar{\beta}(t))_t$ and the data $y^n$/posterior $\pi_n$, the sequence of random variables 
\begin{equation}
\label{eqn:itermediate}
\sup_{t\in [kh,(k+1)h]} \frac{1}{p}\|a(\bar{\beta}(t),t)-a(\bar{\beta}(kh),kh)\|^2
\end{equation}
is bounded in probability as $n\rightarrow \infty$ by
\begin{align}
&\text{p-}\lim_{n\rightarrow \infty}\frac{1}{p}\|a(\bar{\beta}((k+1)h),(k+1)h)-a(\bar{\beta}(kh),kh)\|^2 \nonumber\\
&= \lim_{n\rightarrow \infty} \frac{1}{p}\mathbb{E}[\|a(\bar{\beta}((k+1)h),(k+1)h)-a(\bar{\beta}(kh),kh)\|^2] \label{eqn:deterministic_bound}\\
&\leq \lim_{n\rightarrow \infty}\frac{1}{p}\left(\frac{\|(\frac{1}{\sigma^2}X_{z}^\top X_{z}+kh I)^{-1}\|}{1-\|(\frac{1}{\sigma^2}X_{z}^\top X_{z}+kh I)^{-1} hI\|}\right)^2\mathbb{E}[\|h a(\bar{\beta}(kh),kh)+\bar{\beta}((k+1)h)_{z}-\bar{\beta}(kh)_{z}\|^2]\nonumber\\
&\lesssim \lim_{n\rightarrow \infty}\frac{1}{p}\left(\frac{\|(\frac{1}{\sigma^2}X_{z}^\top X_{z}+kh I)^{-1}\|}{1-h\|(\frac{1}{\sigma^2}X_{z}^\top X_{z}+kh I)^{-1}\|}\right)^2 \mathbb{E}[h^2 \|a(\bar{\beta}(kh),kh)\|^2+\|\bar{\beta}((k+1)h)-\bar{\beta}(kh)\|^2]\nonumber\\
&\lesssim \lim_{n\rightarrow \infty}\frac{1}{p} (h^2 \mathbb{E}[\|a(\bar{\beta}(kh),kh)\|^2]+h\int_{kh}^{(k+1)h} \mathbb{E}[\|a(\bar{\beta}(t),t)\|^2] dt+ph)\nonumber\\
&\lesssim \lim_{n\rightarrow \infty}\frac{1}{p}(h^2 \max_{t\in[kh,(k+1)h]} \mathbb{E}[\|a(\bar{\beta}(t),t)\|^2] + ph) \lesssim h \nonumber
\end{align}
for $h$ sufficiently small such that $h < \lambda_{\min} (X_{z^*}^\top X_{z^*})/\sigma^2$, where we used (1) the update \eqref{eqn:continuous_couple} and Cauchy-Schwarz; (2) $a(\cdot)_j$ is bounded almost surely through the localization process as shown in \cref{lem:alternative}; (3) dominated convergence theorem to exchange limit and expectation together with $\pi(z=z^*|y)\rightarrow 1$ as $n\rightarrow \infty$. Above $\lesssim$ hides constant independent of the dimension $p$.  

The reduction in the first step \eqref{eqn:itermediate} where we go from $\sup$ over $t\in[kh,(k+1)h]$ to $t=(k+1)h$ follows since $t\rightarrow a(\bar{\beta}(t),t)$ is a bounded martingale according to \cref{lem:alternative} for any $a(\cdot)$ constructed with the localization process, therefore $\frac{1}{\sqrt{p}}\|a(\bar{\beta}(t),t)-a(\bar{\beta}(kh),kh)\|$ is a positive bounded sub-martingale for $t\geq kh$ by Jensen's inequality. Then Doob's maximal inequality gives for a fixed $c>0$,
\begin{align*}
&\lim_{n\rightarrow \infty} \mathbb{P}(\sup_{t\in [kh,(k+1)h]} \frac{1}{\sqrt{p}}\|a(\bar{\beta}(t),t)-a(\bar{\beta}(kh),kh)\| \geq c) \\
&\leq \frac{1}{c}\lim_{n\rightarrow \infty} \frac{1}{\sqrt{p}}\mathbb{E}[\|a(\bar{\beta}((k+1)h),(k+1)h)-a(\bar{\beta}(kh),kh)\|]\\
&\leq \frac{1}{c}\lim_{n\rightarrow \infty} \frac{1}{\sqrt{p}}\mathbb{E}[\|a(\bar{\beta}((k+1)h),(k+1)h)-a(\bar{\beta}(kh),kh)\|^2]^{1/2}\, .
\end{align*}
Therefore using \eqref{eqn:deterministic_bound} we can choose $c\gtrsim \sqrt{h}/\epsilon$ deterministically large enough such that the probability above is smaller than $\epsilon$. This in turn implies
\[\text{p-}\lim_{n\rightarrow \infty}\sup_{t\in[kh,(k+1)h]}  \frac{1}{\sqrt{p}}\|a(\bar{\beta}(t),t)-a(\bar{\beta}(kh),kh)\| \lesssim \sqrt{h}\, ,\]
as claimed. 
\end{proof} 

Putting everything together, the theorem below is our main result for the Stochastic Localization sampler. 
\begin{theorem}[Convergence Guarantee for Stochastic Localization Sampler]
\label{thm:sl}
Under the assumptions for \cref{prop:stat}, with probability at least $1-o_n(1)$ over the data and the randomness of the algorithm, for all $kh\leq T$, we have the following recursion for the errors: 
\[\frac{1}{\sqrt{p}}\|\hat{a}(\beta_k,kh)-a(\bar{\beta}(kh),kh)\|\lesssim \frac{1}{kh\sqrt{p}}\|\beta_k-\bar{\beta}(kh)\|+ o_n(1)\lesssim e^{ckh}\sqrt{h}+ o_n(1)\, .\]
Moreover, there is a constant $K$ independent of the dimension such that after $K$ many steps of \cref{alg:SL} where $\mathcal{T}$ is implemented with \cref{lem:mean}, we have $W_2(\pi,\text{Law}(\hat{a}(\beta_K))) \leq \sqrt{p}\zeta$ for any desired tolerance $\zeta$ with probability at least $1-o_n(1)$. The total complexity of the algorithm is $O_p(c^t n^2k)\lesssim O_p(c^k p^3)$ for some constant $c$ if we focus on the scaling with $p$ for warm-start with at most $t\asymp k$ false positives.  
\end{theorem}
\begin{proof}
We couple the continuous $\bar{\beta}(kh)$ and discrete $\beta_k$ processes \eqref{eqn:discrete_update} with the same Brownian increment, i.e.,
\begin{equation}
\label{eqn:continuous_couple}
\bar{\beta}((k+1)h) = \bar{\beta}(kh)+\int_{kh}^{(k+1)h} a(\bar{\beta}(t),t) dt + \int_{kh}^{(k+1)h}dW(t)
\end{equation}
where $\sqrt{h}z_k = \int_{kh}^{(k+1)h}dW(t)$ with same initial condition $\beta_0 = \bar{\beta}(0)=0$ and $a(\beta_0,0)=a(\bar{\beta}(0),0)$. Here $a(\cdot)$ denotes the exact drift from \cref{lem:mean} and $\hat{a}(\cdot)$ the approximate one from \eqref{eqn:approximate_a}. Therefore we have for any $(k+1)h \leq T$, with probability $1-o_n(1)$,
\begin{align*}
&\frac{1}{\sqrt{p}}\|\bar{\beta}((k+1)h)-\beta_{k+1}\| \\
&\leq \frac{1}{\sqrt{p}}\|\bar{\beta}(kh)-\beta_k\|+\frac{1}{\sqrt{p}}\int_{kh}^{(k+1)h}\|a(\bar{\beta}(t),t)-\hat{a}(\beta_k,kh)\|dt\\
&\leq \frac{1}{\sqrt{p}}\|\bar{\beta}(kh)-\beta_k\|+\frac{h}{\sqrt{p}}\|a(\bar{\beta}(kh),kh)-\hat{a}(\beta_k,kh)\|+\frac{h}{\sqrt{p}}\sup_{t\in [kh,(k+1)h]} \|a(\bar{\beta}(t),t)-a(\bar{\beta}(kh),kh)\|\\
&\lesssim \frac{1}{\sqrt{p}}\|\bar{\beta}(kh)-\beta_k\|+\frac{h}{\sqrt{p}}\|a(\bar{\beta}(kh),kh)-\hat{a}(\beta_k,kh)\|+h^{3/2}
\end{align*}
where we used the regularity property from \cref{lem:regular} in the last step. Due to the posterior concentration assumption, using Markov's inequality, with probability $1-o_n(1)$, for any $k$, $\frac{1}{\sqrt{p}}\|\hat{a}(\beta_k)-a(\beta_k)\|\leq \delta (n)$ where $\lim_{n\rightarrow \infty}\delta(n)=0$ is a non-negative deterministic sequence. Together with \cref{lem:regular} give that with probability $1-o_n(1)$,
\begin{align*}
&\frac{1}{\sqrt{p}}\|\hat{a}(\beta_{k+1},(k+1)h)-a(\bar{\beta}((k+1)h),(k+1)h)\| \\
&\leq \frac{1}{\sqrt{p}}\|\hat{a}(\beta_{k+1},(k+1)h)-a(\beta_{k+1},(k+1)h)\|+\frac{1}{\sqrt{p}}\|a(\beta_{k+1},(k+1)h)-a(\bar{\beta}((k+1)h),(k+1)h)\|\\
&\leq \delta(n)+\frac{1}{(k+1)h\sqrt{p}}\|\beta_{k+1}-\bar{\beta}((k+1)h)\|\, .
\end{align*} 
Now putting the last two displays together, and inducting over $k$, we conclude that with high probability
\[\frac{1}{\sqrt{p}}\|\bar{\beta}((k+1)h)-\beta_{k+1}\| \lesssim e^{c(k+1)h}(k+1)h^{3/2} +\delta(n) \, ,\] 
\[\frac{1}{\sqrt{p}}\|\hat{a}(\beta_k,kh)-a(\bar{\beta}(kh),kh)\|\lesssim \frac{1}{kh}(e^{ckh}kh^{3/2}+ \delta(n))+\delta(n)\, ,\]
since it verifies the recursion
\begin{align*}
\frac{1}{\sqrt{p}}\|\bar{\beta}((k+1)h)-\beta_{k+1}\| &\lesssim e^{kh} kh^{3/2} +\delta(n) + \frac{h}{kh} (e^{ckh}kh^{3/2}+ \delta(n))+h\delta(n)+h^{3/2}\\
&\lesssim e^{ckh}h^{3/2}(k+1)+h^{3/2}+\delta(n)\\
&\lesssim e^{c(k+1)h}(k+1)h^{3/2} +\delta(n)\, ,
\end{align*}
finishing the first part of the statement. 

This in turn implies using the continuous time convergence rate from \cref{lem:continuous_time} and the coupling definition of the $W_2$ distance, for $K=T/h$,
\begin{align*}
\frac{1}{\sqrt{p}}W_2(\pi,\text{Law}(\hat{a}(\beta_K))) &\leq \frac{1}{\sqrt{p}}W_2(\pi,\text{Law}(a(\bar{\beta}(T)))) + \frac{1}{\sqrt{p}}W_2(\text{Law}(a(\bar{\beta}(T))),\text{Law}(\hat{a}(\beta_K)))\\
&\leq 1/\sqrt{T}+C(T)h^{1/2}+\delta(n)
\end{align*}
therefore for $n$ sufficiently large, when $T$ is sufficiently large and $h$ suitably small (both are independent of the dimension), we have $W_2(\pi,\text{Law}(\hat{a}(\beta_K))) \leq \sqrt{p}\zeta$, for any desired $\zeta>0$, which holds with probability $1-o_n(1)$ w.r.t randomness in $y$ such that \eqref{eqn:high_prob_selection} holds ($X$ deterministically verifies restricted eigenvalue properties). The complexity of the algorithm now follows by putting together with \cref{lem:mean}.
\end{proof}

The main benefit of the Stochastic Localization sampler lies in its obliviousness to the ``ill-design" of the data matrix $X$ (e.g., if there are strong correlation between some columns of $X$), where we see from \cref{thm:gibbs} that even under warm-start and posterior contraction ($s=0$), such terms still show up and scale with the mixing time \emph{exponentially}. The guarantee of \cref{thm:sl} is in $W_2$ distance and not TV, but both have $\sqrt{p}$ scaling with dimension. In both cases the scaling with the number of initial false positives $t$ is less than ideal, but a warm-start is essentially necessary for efficiently simulating from such a mixture posterior.

\section{(Frequentist) Statistical Properties of Posterior \eqref{eqn:quasi_posterior}}
\label{sec:statistics}
In this section, we justify the posterior concentration assumption made on the sparsified likelihood model \eqref{eqn:quasi_posterior}. We highlight the importance of \emph{diffusing and shrinking} priors for this class of posteriors as in \cite{narisetty2014bayesian} (i.e., allowing the prior parameters to depend on $n$), which is required for strong model selection consistency $\pi(z=z^*|y )\xrightarrow{P} 1$ as $n\rightarrow \infty$ in high-dimensional setting where $p$ is allowed to grow with $n$ exponentially, i.e., $p_n=e^{o(n)}$. This choice can in some sense be seen as adjusting for multiplicity.

\subsection{Warm-up: Sparse Normal Means Model}
\label{sec:stat_normal_mean}
Let us motivate the choice of $\tau_0,\tau_1, q$ by considering the setup $X^\top X = n I_p$ where $p_n\leq n$, and study under what conditions on the priors do the corresponding posteriors confer model selection consistency. 
\begin{lemma}
With a sparsified likelihood, the model selection consistency requirement is the same for point-mass spike and Gaussian spike under orthogonal design, which is satisfied by the choice in \cref{assume:scaling} under $\beta$-min condition \eqref{eqn:beta_min}. 
\end{lemma}
\begin{proof}
The posterior for $z$ is (define $\hat{\beta}_j:= y^\top X_{z^*,j}/n$)
\begin{align*}
&\mathbb{P}(z=z^*|y) \\
&\propto \int_{\mathbb{R}^p} \exp(-\frac{1}{2\sigma^2}\|y-X_{z^*}\beta_{z^*}\|^2)\prod_{j=1}^p(\frac{1-q}{\tau_0}\exp(-\frac{\beta_j^2}{2\tau_0^2}))^{1-z_j^*}(\frac{q}{\tau_1}\exp(-\frac{\beta_j^2}{2\tau_1^2}))^{z_j^*} d\beta\\
&\propto \prod_{z^*_j=0} \int_{\mathbb{R}}  (\frac{1-q}{\tau_0}\exp(-\frac{\beta_j^2}{2\tau_0^2}))^{1-z_j^*}d\beta_j\prod_{z_j^*=1} \int_{\mathbb{R}} \exp(-\frac{n}{2\sigma^2}(\beta_j-\hat{\beta}_j)^2)(\frac{q}{\tau_1}\exp(-\frac{\beta_j^2}{2\tau_1^2}))^{z_j^*}  d\beta_j\\
&= \prod_{z_j^*=1} \frac{\mathbb{E}_{\beta_j\sim \mathcal{N}(\hat{\beta}_j,\sigma^2/n)}[\exp(-\beta_j^2/2\tau_1^2)]\cdot \frac{q}{\tau_1}\exp(\frac{1}{2\sigma^2 n}y^\top X_{z^*,j}X_{z^*,j}^\top y)}{\mathbb{E}_{\beta_j\sim \mathcal{N}(\hat{\beta}_j,\sigma^2/n)}[\exp(-\beta_j^2/2\tau_1^2)]\cdot \frac{q}{\tau_1}\exp(\frac{1}{2\sigma^2 n}y^\top X_{z^*,j}X_{z^*,j}^\top y)+1-q} \times\\
&\prod_{z_j^*=0} \frac{1-q}{\mathbb{E}_{\beta_j\sim \mathcal{N}(\hat{\beta}_j,\sigma^2/n)}[\exp(-\beta_j^2/2\tau_1^2)]\cdot \frac{q}{\tau_1}\exp(\frac{1}{2\sigma^2 n}y^\top X_{z^*,j}X_{z^*,j}^\top y)+1-q}\\
&=:  \prod_{z_j^*=1} a_j \prod_{z_j^*=0} b_j
\end{align*}
One can show for each of the $k$ terms corresponding to $z_j^*=1$ using completion of squares,
\begin{align*}
r_j &:=\frac{q}{\tau_1}\mathbb{E}_{\beta_j\sim \mathcal{N}(\hat{\beta}_j,\sigma^2/n)}[\exp(-\beta_j^2/2\tau_1^2)]\exp(\frac{1}{2\sigma^2 n}y^\top X_{z^*,j}X_{z^*,j}^\top y)\\
&=\frac{q}{\sqrt{1+\frac{n \tau_1^2}{\sigma^2}}}\exp(\frac{1}{2}(\frac{1}{\tau_1^2}+\frac{n}{\sigma^2})^{-1}\frac{\hat{\beta}_j^2n^2}{\sigma^4}-\frac{\hat{\beta}_j^2n}{2\sigma^2})\exp(\frac{\hat{\beta}_j^2 n}{2\sigma^2})\\
&= \frac{q}{\sqrt{1+\frac{n \tau_1^2}{\sigma^2}}}\exp(\frac{1}{2}(\frac{1}{\tau_1^2}+\frac{n}{\sigma^2})^{-1}\frac{(y^\top X_{z^*,j})^2}{\sigma^4})
\end{align*}
where we recall $y=X_{z^*}\beta^*_{z^*}+\epsilon$ and we require for $j$ such that $z_j^*=1$, $|\hat{\beta}_j|^2 > \frac{c\sigma^2\log(p)}{n}$ for a large enough $c$, and $\|z^*\|_0=k_n <p_n \leq n$.  

To have $\mathbb{P}(z=z^*|y )\xrightarrow{P} 1$, a sufficient condition is to have $\sum_{j=1}^p \mathbb{P}(z_j\neq z_j^*|y) \xrightarrow{P} 0$, or equivalently, $\min_{j\in [p]} \mathbb{P}(z_j= z_j^*|y) \geq 1-\frac{\eta}{p}$ for a sufficiently small $\eta$; but generally requiring $\min_{j\in [p]} \mathbb{P}(z_j= z_j^*|y)\xrightarrow{P} 1$ is a weaker consistency result. We begin with the first term, for the product of $k$ terms to go to $1$ as $n\rightarrow \infty$, we see that using Bernoulli's inequality,
\[1\leftarrow \prod_{z_j^*=1} a_j \geq ( \min_{z_j^*=1} a_j)^{k_n} \geq (1-\max_{z_j^*=1}  \mathbb{P}(z_j=0|y))^{k_n}\geq 1-k_n \cdot\max_{z_j^*=1}  \mathbb{P}(z_j=0|y)\]
therefore we need $k_n \cdot\max_{z_j^*=1} \mathbb{P}(z_j=0|y) \rightarrow 0$, which means it suffices for $|\hat{\beta}_j|^2 \asymp \frac{\sigma^2\log(p)}{n}$
\[\frac{1-q}{r_j}=\frac{1-q}{q}\sqrt{1+\frac{n\tau_1^2}{\sigma^2}}\exp(-\frac{1}{2}(\frac{1}{\tau_1^2}+\frac{n}{\sigma^2})^{-1}\frac{\hat{\beta}_j^2n^2}{\sigma^4}) \ll \frac{1}{k_n}\, .\]
Similarly for the second term, 
\[1\leftarrow \prod_{z_j^*=0} b_j \geq ( \min_{z_j^*=0} b_j)^{p} \geq (1-\max_{z_j^*=0}  \mathbb{P}(z_j=1|y))^{p}\geq 1-p \cdot\max_{z_j^*=0}  \mathbb{P}(z_j=1|y)\, .\]
which implies since $z_j^*=0$, the $\exp(\cdot)$ term from $r_j$ vanishes using \eqref{eqn:beta_min},
\[ \frac{q}{(1-q)\sqrt{1+\frac{n \tau_1^2}{\sigma^2}}} \ll \frac{1}{p}\, .\]
In both cases, it suffices to impose $(1-q)/q\sim p,\tau_1\sim \sigma p/\sqrt{n}$, and it is crucial for them to scale with $(n,p)$ to achieve model selection consistency. In this particular case $\tau_0$ doesn't play a role, e.g., whether we pick point-mass spike or Gaussian spike.  
\end{proof} 

Variable selection is generally considered a harder problem than parameter estimation / prediction \cite{buhlmann_geer_book}, therefore one should expect good performance with respect to those criteria as well from these choices, which is indeed the case as shown next in the regression setting.

\subsection{Posterior Contraction}
 \label{subsec:posterior_contract}
We show that posterior contraction conditions in $\mathcal{E}_s$ are statistically grounded in this section. From an information-theoretic perspective, one generally needs some identifiability assumptions on the design matrix for statistical estimation / posterior consistency, and these will show up here as follows.
\begin{hypothesis}
\label{assume:re_relax}
For all $u\in \mathbb{R}^p$ such that $\|z^{*,c}(u-\beta^*)\|_1 \leq 7\|z^*(u-\beta^*)\|_1$, there exist $R>0$ where for $\|\beta^*\|_0=k$,
\[\frac{1}{n}(u-\beta^*)^\top(X^\top X) (u-\beta^*)\geq R\cdot \|z^*(u-\beta^*)\|_2^2\, ,\]
which is closely related to \eqref{eqn:re}, but slightly relaxed so the restricted eigenvalue direction can be not exactly sparse, rather take small values off support of $z^*$. We also assume a general condition for $k'\asymp k$, the exact-sparsity restricted eigenvalue $\min_{\|v\|_0\leq k'} v^\top (X^\top X)v \geq \omega(k') n \|v\|^2$ is bounded away from $0$ by a small constant $\omega(k')$. Additionally, $\|\beta^*\|_\infty = \mathcal{O}(1)$ doesn't grow with $n$. 
\end{hypothesis}
We consider a sequence of problems with $n,p\rightarrow \infty$ where $n=o(p)$ and demonstrate that for $s=0$, condition
\[\pi(z\in\{0,1\}^p\colon z^*\subset z, \|z\|_0\leq \|z^*\|_0+s\vert y)\geq 1-\frac{4}{p^{\frac{\delta}{2}(s+1)}}\]
from \Cref{sec:mixing_gibbs} holds with $\delta=2$, which implies $\pi(z^*\vert y)\geq 1/2$ for $p>9$ with probability at least $1-o_n(1)$ over the data, since using (3) from \cref{prop:stat} and Markov's inequality
\begin{equation}
\label{eqn:high_prob_selection}
\mathbb{P}((1-\pi(z^*|y))\geq 1/p) \leq \frac{\mathbb{E}[1-\pi(z^*|y)]}{1/p}\leq \frac{1/p^2}{1/p}=\frac{1}{p}\, .
\end{equation}
We characterize the large scale behavior of the posterior \eqref{eqn:quasi_posterior} below.

\begin{proposition}[Frequentist Guarantee on the Posterior]
\label{prop:stat}
Under the parameter choice $q/(1-q)\sim 1/p^{\delta+1}$ for some constant $\delta>0$, $\tau_1\sim \sigma p/\sqrt{n}, \|X_j\|_2^2=n$ from \cref{assume:scaling}, in addition to the $\beta$-min condition \eqref{eqn:beta_min} and \cref{assume:re_relax} above, it holds in the regime $p_n=e^{o(n)}$ that 
\begin{enumerate}
\item $\mathbb{E}\left[\pi(z:\|z\|_0\gtrsim k(1+1/\delta)|y)\right] \leq \frac{2}{p^2}$
\item $\mathbb{E}[\pi(B^c | y)] \lesssim 1/p^2$ for $B=\cup_{z:\|z\|_0\lesssim k}\; \{\beta:\|\beta_z-\beta^*\|\lesssim \frac{\sigma \sqrt{k\log(p)}}{\sqrt{n}\omega(k)}, \|\beta_z-\beta\|\lesssim \tau_0 \sqrt{p}\}$
\item $\mathbb{E}[\pi(z^*|y)] \gtrsim 1-\frac{1}{p^2}$
\end{enumerate}
where in the above expectation is taken with respect to the noise $\epsilon$ only and $X$ deterministically satisfy the stated assumptions. Moreover, for \cref{assume:re_relax} to hold with probability tending to one as $n\rightarrow \infty$, for example with a Gaussian design matrix $X_{ij}\sim \mathcal{N}(0,1)$, it entails $n\gtrsim k\log(p)$ and $k\lesssim \log(p)$ for the sample size and sparsity level respectively. In general $n,k$ will scale with the ``ill-design-ness" of the matrix $X$.
\end{proposition}
\begin{proof}
We build upon the result in \cite{atchade2018approach} and verify the conditions stated there. In our case, $\ell(\beta_{z},y)=\frac{1}{2\sigma^2}\|y-X_{z}\beta_{z}\|^2$, therefore with probability at least $1-2/p^2$ since $\epsilon\sim \mathcal{N}(0,\sigma^2 I)$,
\begin{equation}
\label{eqn:event_H}
\|\nabla \ell(\beta^*;y)\|_\infty=\|-\frac{1}{\sigma^2}X^\top (y-X\beta^*)\|_{\infty}=\frac{1}{\sigma^2}\|X^\top\epsilon\|_\infty \leq \frac{\sqrt{n}}{\sigma}\sqrt{2\log(p)}=: \frac{\bar{\rho}}{2}
\end{equation}
and for $\beta,\beta^* \in \mathbb{R}^p$ where $\beta$ has the same support as $\beta^*$, since $X_{z}^\top X_{z}\preceq \|X_{z}\|_F^2 \cdot I$,
\[\mathcal{L}_{\beta^*}(\beta;y) = -\frac{1}{2\sigma^2}(\beta-\beta^*)^\top (X^\top X) (\beta-\beta^*) \geq -\frac{nk}{2\sigma^2}\|\beta-\beta^*\|_2^2 =: -\frac{\bar{\kappa}}{2}\|\beta-\beta^*\|_2^2\]
which means \textbf{H1} is satisfied. For \textbf{H2}, it suffices to check $p^{\delta/2}\gtrsim e^{2n/\sigma^2 p^2}$, which holds under $p=e^{o(n)}$ as we assume. Starting from Theorem 2 therein, we check equation (2), picking $\mathcal{E}$ as the intersection of \eqref{eqn:event_H} and \cref{assume:re_relax}, we have on this event using the Gaussian moment-generating function, 
\begin{align*}
\mathbb{E}[e^{\mathcal{L}_{\beta^*}(\beta;y)+(1-\frac{\rho_1}{\bar{\rho}})\langle \nabla \ell(\beta^*;y),\beta-\beta^*\rangle}] &=\mathbb{E}[e^{-\frac{1}{2\sigma^2}(\beta-\beta^*)^\top (X^\top X) (\beta-\beta^*)-\frac{1-\rho_1/\bar{\rho}}{\sigma^2}(\beta-\beta^*)^\top X^\top \epsilon}] \\
&=\mathbb{E}[e^{-\frac{1}{2\sigma^2}(1-(1-\rho_1/\bar{\rho})^2)(\beta-\beta^*)^\top (X^\top X) (\beta-\beta^*)}]\\
&\leq e^{-\frac{Rn(1-(1-\rho_1/\bar{\rho})^2)}{2\sigma^2}\|\beta-\beta^*\|_2^2}
\end{align*}
therefore we can pick the rate function $r_0(x)=\frac{Rn(1-(1-\rho_1/\bar{\rho})^2)}{\sigma^2} x^2$ for such $\beta$'s. Since $\rho_1$ in our context is $1/\tau_1^2$, it is clear that $\rho_1 < \bar{\rho}$, therefore $r_0(x) \geq \frac{Rn}{\sigma^2\tau_1^2 \bar{\rho}}x^2, x\geq 0$, which means since neither $R$ nor $\sigma$ scales with $n$, 
\[a_0:= -\min_{x>0}\left\{ r_0(x)-\frac{4}{\tau_1^2}\sqrt{k}x\right\}\lesssim \frac{k\sqrt{n\log(p)}}{R\sigma p^2}\]
is bounded above by an absolute constant. It can then be checked that condition (3):
\[k(\frac{1}{2}+\frac{2}{\tau_1^2})+\frac{k}{2}\log(1+\bar{\kappa}\tau_1^2)+\frac{a_0}{2}+\frac{2}{\tau_1^2}\|\beta^*\|_2^2 \leq c_0 k\log(p)\]
holds with an absolute constant $c_0$ with the specified $\tau_1$ and $\|\beta^*\|_\infty$. Therefore Theorem 2 concludes for $k':=k(1+\frac{1}{\delta}) > k$, picking $j= 4/\delta$,
\begin{equation}
\label{eqn:sparse}
\mathbb{E}\left[\pi(z:\|z\|_0\gtrsim k'|y)\right] \leq \frac{2}{p^2}\, .
\end{equation}
For the second part, again using \cref{assume:re_relax}, for all $\beta$ with at most $k'$ active coordinates,
\[\mathcal{L}_{\beta^*}(\beta;y)=-\frac{1}{2\sigma^2}(\beta-\beta^*)^\top (X^\top X) (\beta-\beta^*) \leq -\frac{1}{2}\frac{n}{\sigma^2}\omega(k+k')\|\beta-\beta^*\|_2^2=: -\frac{1}{2}r(\|\beta-\beta^*\|_2)\]
therefore we are on the event $\mathcal{E}_1(k')$ with the above rate function. Take the contraction radius 
\begin{align*}
\zeta &:=\inf\left\{l>0: \frac{n}{\sigma^2}\omega(k+k')x^2-4\sqrt{k+k'}\frac{\sqrt{n\log(p^2)}}{\sigma}x \geq 0 \quad\forall x\geq l\right\}\\
&\asymp \frac{\sigma\sqrt{(k'+k)\log(p^2)}}{\sqrt{n}\omega(k+k')} \asymp \frac{\sigma \sqrt{k\log(p)}}{\sqrt{n}\omega(k)}
\end{align*}
Note this contraction rate is largely comparable to the ``ideal" near-minimax benchmark in \eqref{eqn:gold-standard} assuming $\omega(k)$ is a constant. Now we check that equation (8) 
\[C\frac{\sqrt{n\log(p^2)}}{\sigma}\sqrt{k+k'}\frac{\sigma\sqrt{(k'+k)\log(p^2)}}{\sqrt{n}\omega(k+k')} \gtrsim \max\{k'\log(p),(1+\delta)k\log(p+p^3 k)\}\]
holds with an absolute constant $C$ since we assume both $\omega(k+k')$ and $\delta$ to be constants. Applying Theorem 3, together with \eqref{eqn:sparse} gives the contraction rate
\begin{align}
\label{eqn:contract}
\mathbb{E}[\pi(B^c | y)]\leq \frac{2}{p^2}+8e^{-\frac{\sqrt{n\log(p^2)}}{\sigma}\sqrt{k+k'}\zeta} +2e^{-p}\lesssim \frac{1}{p^2}
\end{align}
 where we define the set
\[B:=\cup_{z:\|z\|_0\leq k'}\; \{\beta:\|\beta_z-\beta^*\|\lesssim \zeta, \|\beta_z-\beta\|\lesssim \tau_0 \sqrt{p}\}\, ,\]
which describes the set of $\beta$'s that have most of the mass concentrated on $k'$-sparse sub-vector and on the support is close to $\beta^*$.

Now for the (perfect) model selection, on event $\mathcal{E}_2(k')$ we have
\[
 \cap_{j=1}^{k'-k}\; \mathcal{U}_j := \cap_{j=1}^{k'-k} \left\{\max_{z^*\subset z,\|z\|_0=k+j}\; \frac{1}{2\sigma^2}(\|y-X\beta_{z}\|_2^2 - \|y-X\beta_{z^*}\|_2^2)\leq \frac{j\delta}{2}\log(p)\right\}\, ,
\]
which happens with high probability since by union bound and $\|y-X\beta_{z}\|^2=\|(I-P_z)y\|^2=\|y\|^2-\|P_z y\|^2$,
\begin{align}
\label{eqn:chi-square-set}
&\sum_{j=1}^{k'-k}\mathbb{P}(\mathcal{U}_j^c) =  \sum_{j=1}^{k'-k} \mathbb{P}\left(\max_{z^*\subset z,\|z\|_0=k+j} y^\top (P_{z^*}-P_z)y \geq j\delta\sigma^2 \log(p)\right) \\
&= \sum_{j=1}^{k'-k} \mathbb{P}\left(\max_{z^*\subset z,\|z\|_0=k+j} \chi^2(\text{dof}=\|z\|_0-\|z^*\|_0,\text{non-central}=(X\beta^*)^\top(P_{z^*}-P_z)X\beta^*) \geq j\delta\sigma^2 \log(p)\right)\nonumber\\
&\lesssim \sum_{j=1}^{k'-k} p^{-\frac{\sigma^2 \delta j}{4}} \lesssim \frac{1}{p^2} \nonumber
\end{align}
where we used the concentration inequality for the central $\chi^2$ distribution since $y=X\beta^*+\epsilon \sim \mathcal{N}(X\beta^*,\sigma^2 I)$, the above non-centrality parameter is in fact $0$ and $P_z\in \mathbb{R}^{n\times n}$ denotes the orthogonal projector onto the column span of $X_{z}$ (idempotent of rank $\|z\|_0$), and similarly for $P_{z^*}$. We also used that $z^*\subset z$ above. 

We can also deduce that $\bar{\kappa}=\frac{nk}{\sigma^2}, \underline{\kappa}=\frac{n\omega(k+k')}{\sigma^2}$, i.e., the matrix $X_{z}$ is full-column rank (restricted strong-convexity) and restricted smooth on the event $\mathcal{E}_1(k')$ (since Hessian is constant, the inner $\inf$ and $\sup$ in the definition of (12) and (13) are immaterial here). Invoking Theorem 5 by setting $j=0$ with $a_2=0$ since $\ell$ is quadratic, with the $\beta$-min condition \eqref{eqn:beta_min} yields 
\[\mathbb{E}\left[\mathbbm{1}\{\cap_{j=1}^{k'-k}\; \mathcal{U}_j \} (1-\pi(z^*|y))\right] \lesssim e^{\frac{\sqrt{k'}\zeta}{\tau_1^2}}\sqrt{\frac{1}{\tau_1^2 \underline{\kappa}p^\delta}}+\frac{1}{p^2}\lesssim \frac{1}{p^2} \, ,\]
where we used \eqref{eqn:contract} and $\underline{\kappa}p^\delta \gtrsim 1/\tau_1^2$ that is satisfied by our choice. Now to remove the conditional event inside, putting together with \eqref{eqn:chi-square-set} gives the desired result
\[\mathbb{E}[\pi(z^*|y)] \gtrsim 1-\frac{1}{p^2}\, ,\]
since 
\begin{align*}
\mathbb{E}\left[1-\pi(z^*|y)\right] &\leq \mathbb{E}[\mathbbm{1}\{\cap_{j=1}^{k'-k}\; \mathcal{U}_j\} (1-\pi(z^*|y))] +\mathbb{P}(\{\cap_{j=1}^{k'-k}\; \mathcal{U}_j\}^c)\\
&\leq \mathbb{E}[\mathbbm{1}\{\cap_{j=1}^{k'-k}\; \mathcal{U}_j\} (1-\pi(z^*|y))] +\sum_{j=1}^{k'-k}\mathbb{P}(\mathcal{U}_j^c),
\end{align*}
and the last required condition $\zeta\sqrt{\underline{\kappa}}\gtrsim \sqrt{k}$ also checks out.

The last claim about the scaling of $n,k$ for the Gaussian design to hold with high probability follows from well-known results in high-dimensional statistics \cite{buhlmann_geer_book} -- the condition on $\omega(k')$ is already used in the proof of \cref{lem:warm-start}. 
%
\end{proof}

This result implies that in the high-dimensional regime $n=o(p)$ and for \emph{well-chosen} parameters, one has with high probability (1) sparse support; (2) contraction towards $\beta^*$; (3) model selection consistency for the posterior $\pi(\cdot|y)$. We remark that the result does not in fact depend crucially on the scaling of $\tau_0$ (the prior for the spike), other than it should decrease with $n$. Both the posterior contraction rate and the dependence of prior parameters on $n,p$ also bear resemblance with another family of continuous priors \cite[Theorem 6.4]{spike-and-slab-lasso} with heavier-tailed Laplace spike and slab, assuming $q$ fixed (i.e., non-hierarchical prior). 

\begin{remark}
In fact, the relative density ratio expression from \eqref{eqn:relative_density_1}-\eqref{eqn:relative_density_2} also hint at a connection to $\ell_0$-penalty if we look at the posterior mode. Since we have $\tau_1 \rightarrow \infty$ and $q/(1-q)\sim 1/p$, 
\begin{align*}
&\arg\max_{z\in\mathcal{E}_s} \log\left(\frac{\pi(z\vert y)}{\pi(z^*\vert y)}\right) \\
&=\arg\max_{z\in\mathcal{E}_s} \log\left( \frac{(\frac{q}{1-q})^{\|z\|_0-\|z^*\|_0}}{\sqrt{\det(I+\frac{\tau_1^2}{\sigma^2}X_{z-z^*}^\top (I+\frac{\tau_1^2}{\sigma^2}X_{z^*}X_{z^*}^\top)^{-1}X_{z-z^*})}}   \frac{\exp(-\frac{\tau_1^2}{2\sigma^2}y^\top X_z (\tau_1^2 X_z^\top X_z+\sigma^2 I)^{-1} X_z^\top y)}{\exp(-\frac{\tau_1^2}{2\sigma^2}y^\top X_{z^*} (\tau_1^2 X_{z^*}^\top X_{z^*}+\sigma^2 I)^{-1} X_{z^*}^\top y)}\right)\\
&\approx \arg\min_{z\in\mathcal{E}_s} \, (\|z\|_0-\|z^*\|_0)\log(p)+\frac{1}{2\sigma^2}(\|X\beta_z-y\|^2-\|X\beta_{z^*}-y\|^2)\\
&\quad\quad \quad +\log\left(\sqrt{\det(I+X_{z-z^*}^\top(X_{z^*}X_{z^*}^\top)^{-1}X_{z-z^*})}\right)\\
&\approx \arg\min_{z\in\mathcal{E}_s} \, (\|z\|_0-k)\log(p)+\frac{1}{2\sigma^2}\|X\beta_z-y\|^2\, ,  
\end{align*}
which means that asymptotically when the posterior concentrates on $z\in\mathcal{E}_s$ with $\leq s$ false positives, since \eqref{eqn:re} implies the $\det(\cdot)$ is uniformly bounded away from $0$ on this set, the posterior mode is approximately imposing a $\ell_0$-penalty on the model size while trading off with data fitting.
\end{remark}

The following is an immediate corollary that shows the posterior spread can quantify the remaining uncertainty for inferring $\beta^*$ based on the observed data $y^n$. Note $\mathcal{C}_n(y^n)$ below is random since it's constructed using the data $y^n$. We omit the proof as it is straightforward.
\begin{corollary}
Given the conditions that allow consistent model selection $\pi(z=z^*|y )\xrightarrow{P} 1$, credible sets for individual parameters $\beta_j$ building upon the posterior are valid asymptotic confidence sets: $\pi_n(\mathcal{C}_n|y^n)=1-\alpha \Rightarrow \mathbb{P}_{\beta^*}(\beta_j^*\in \mathcal{C}_n)\xrightarrow{n\rightarrow \infty} 1-\alpha$ by virtue of the BvM distributional approximation from \cite[Theorem 7]{atchade2018approach} and equation (15) therein.
\end{corollary}

We also mention in passing that the fact we assumed $\epsilon\sim \mathcal{N}(0,\sigma^2 I)$ should not be considered a limitation for the statistical guarantee stated above. For example, for $\epsilon$ with subgaussian tails, concentration inequality for the quadratic form \eqref{eqn:chi-square-set} and \eqref{eqn:event_H} are readily available. Therefore the posterior \eqref{eqn:quasi_posterior}, which would be slightly mis-specified in this case, is still a meaningful object for inference and design sampling procedures for.

\section{Discussion}
Our work contributes to the ongoing effort of understanding statistical / computational trade-offs arising from contemporary data science problems. The continuous spike-and-slab priors with quasi-likelihood we study strike good balance between these two goals. While the number of submodels scales as $2^p$, natural statistical considerations indicate that it is not necessary to explore the entire state space to get a good approximate sample from the posterior for inference purpose. Moreover, under the same (1) posterior concentration on the parameter; and (2) warm start conditions (possibly implemented using a frequentist point estimator) that enable efficient sampling with a Gibbs sampler, we propose an improved method, based on Stochastic Localization, that is oblivious to the well-posedness of the design matrix. 

Much like the flurry of work on non-convex optimization which demonstrate that, under various mild statistical assumptions on the data/model and with possibly good initialization, simple gradient-based method can be shown to find good local/global minima efficiently; what we observe in this work is similar in spirit for the sampling analogue that exploit problem structure to avoid worst-case scenarios for sampling from non-log-concave distributions. Beyond spike-and-slab models, the Stochastic Localization sampler can be more broadly applicable whenever an estimate of the denoising drift $\mathbb{E}[\beta|\theta_t=\theta]$ is available (not necessarily in closed-form, an output from an efficient algorithm is also an option) for the Gaussian estimation problem \eqref{eqn:snr}, which can be especially useful when the posterior arising from interesting Bayesian statistical models exhibit multi-modal structure -- they pose challenge for MCMC-based method but seem to be quite prevalent in practice.




\bibliographystyle{siamplain}
\bibliography{main}

\begin{thebibliography}{10}

\bibitem{atchade2018approach}
{\sc Y.~Atchade and A.~Bhattacharyya}, {\em {An approach to large-scale
  Quasi-Bayesian inference with spike-and-slab priors}}, arXiv preprint
  arXiv:1803.10282,  (2018).

\bibitem{atchade2021approximate}
{\sc Y.~F. Atchad{\'e}}, {\em {Approximate spectral gaps for Markov chain
  mixing times in high dimensions}}, SIAM Journal on Mathematics of Data
  Science, 3 (2021), pp.~854--872.

\bibitem{belloni2009computational}
{\sc A.~Belloni and V.~Chernozhukov}, {\em {On the computational complexity of
  MCMC-based estimators in large samples}}, The Annals of Statistics, 37
  (2009), pp.~2011 -- 2055, \url{https://doi.org/10.1214/08-AOS634}.

\bibitem{fast_conditional_gaussian}
{\sc A.~Bhattacharya, A.~Chakraborty, and B.~K. Mallick}, {\em {Fast sampling
  with Gaussian scale mixture priors in high-dimensional regression}},
  Biometrika, 103 (2016), pp.~985--991,
  \url{https://doi.org/10.1093/biomet/asw042},
  \url{https://doi.org/10.1093/biomet/asw042}.

\bibitem{biswas2022scalable}
{\sc N.~Biswas, L.~Mackey, and X.-L. Meng}, {\em {Scalable Spike-and-Slab}}, in
  International Conference on Machine Learning, PMLR, 2022, pp.~2021--2040.

\bibitem{buhlmann_geer_book}
{\sc P.~B{\"u}hlmann and S.~Van De~Geer}, {\em {Statistics for High-Dimensional
  Data: Methods, Theory and Applications}}, Springer Science \& Business Media,
  2011.

\bibitem{castillo2015regression}
{\sc I.~Castillo, J.~Schmidt-Hieber, and A.~Van~der Vaart}, {\em {Bayesian
  linear regression with sparse priors}}, The Annals of Statistics, 43 (2015),
  pp.~1986--2018.

\bibitem{castillo_sequence}
{\sc I.~Castillo and A.~van~der Vaart}, {\em {Needles and Straw in a Haystack:
  Posterior concentration for possibly sparse sequences}}, The Annals of
  Statistics, 40 (2012), pp.~2069 -- 2101,
  \url{https://doi.org/10.1214/12-AOS1029},
  \url{https://doi.org/10.1214/12-AOS1029}.

\bibitem{yuansi2021KLS}
{\sc Y.~Chen}, {\em {An almost constant lower bound of the isoperimetric
  coefficient in the KLS conjecture}}, Geometric and Functional Analysis, 31
  (2021), pp.~34--61.

\bibitem{prox}
{\sc Y.~Chen, S.~Chewi, A.~Salim, and A.~Wibisono}, {\em {Improved analysis for
  a proximal algorithm for sampling}}, in Conference on Learning Theory, PMLR,
  2022, pp.~2984--3014.

\bibitem{donoho_threshold}
{\sc D.~L. Donoho and I.~M. Johnstone}, {\em {Ideal spatial adaptation by
  wavelet shrinkage}}, Biometrika, 81 (1994), pp.~425--455,
  \url{https://doi.org/10.1093/biomet/81.3.425}.

\bibitem{alaoui2022sampling}
{\sc A.~El~Alaoui, A.~Montanari, and M.~Sellke}, {\em {Sampling from the
  Sherrington-Kirkpatrick Gibbs measure via algorithmic stochastic
  localization}}, in 2022 IEEE 63rd Annual Symposium on Foundations of Computer
  Science (FOCS), IEEE, 2022, pp.~323--334.

\bibitem{eldan_note}
{\sc R.~Eldan}, {\em {From stochastic calculus to geometric inequalities}},
  Lecture Notes.

\bibitem{shotgun}
{\sc C.~Hans, A.~Dobra, and M.~West}, {\em {Shotgun Stochastic Search for
  “Large p” Regression}}, Journal of the American Statistical Association,
  102 (2007), pp.~507--516, \url{https://doi.org/10.1198/016214507000000121}.

\bibitem{klartag2021bayesian}
{\sc B.~Klartag and E.~Putterman}, {\em {Spectral monotonicity under Gaussian
  convolution}}, arXiv preprint arXiv:2107.09496,  (2021).

\bibitem{narisetty2014bayesian}
{\sc N.~N. Narisetty and X.~He}, {\em {Bayesian variable selection with
  shrinking and diffusing priors}}, The Annals of Statistics, 42 (2014),
  pp.~789 -- 817, \url{https://doi.org/10.1214/14-AOS1207}.

\bibitem{narisetty2018skinny}
{\sc N.~N. Narisetty, J.~Shen, and X.~He}, {\em {Skinny gibbs: A consistent and
  scalable Gibbs sampler for model selection}}, Journal of the American
  Statistical Association,  (2018).

\bibitem{nickl}
{\sc R.~Nickl}, {\em Bayesian non-linear statistical inverse problems}, Lecture
  Notes ETH Zurich,  (2022).

\bibitem{polya-gamma}
{\sc N.~G. Polson, J.~G. Scott, and J.~Windle}, {\em {Bayesian inference for
  logistic models using P{\'o}lya--Gamma latent variables}}, Journal of the
  American statistical Association, 108 (2013), pp.~1339--1349.

\bibitem{Raginsky}
{\sc M.~Raginsky}, {\em {Sampling Using Diffusion Processes, from Langevin to
  Schr{\"o}dinger}}, Notes,  (2021).

\bibitem{ray2022variational}
{\sc K.~Ray and B.~Szab{\'o}}, {\em {Variational Bayes for high-dimensional
  linear regression with sparse priors}}, Journal of the American Statistical
  Association, 117 (2022), pp.~1270--1281.

\bibitem{rovckova2014emvs}
{\sc V.~Ro{\v{c}}kov{\'a} and E.~I. George}, {\em {EMVS: The EM approach to
  Bayesian variable selection}}, Journal of the American Statistical
  Association, 109 (2014), pp.~828--846.

\bibitem{spike-and-slab-lasso}
{\sc V.~Ro{\v{c}}kov{\'a} and E.~I. George}, {\em {The spike-and-slab Lasso}},
  Journal of the American Statistical Association, 113 (2018), pp.~431--444.

\bibitem{handbook2021}
{\sc M.~G. Tadesse and M.~Vannucci}, {\em {Handbook of Bayesian variable
  selection}}, CRC Press, 2021.

\bibitem{schrodinger}
{\sc F.~Vargas, A.~Ovsianas, D.~Fernandes, M.~Girolami, N.~D. Lawrence, and
  N.~N{\"u}sken}, {\em {Bayesian learning via neural
  Schr{\"o}dinger--F{\"o}llmer flows}}, Statistics and Computing, 33 (2023),
  p.~3.

\bibitem{yang2016computational}
{\sc Y.~Yang, M.~J. Wainwright, and M.~I. Jordan}, {\em {On the computational
  complexity of high-dimensional Bayesian variable selection}}, The Annals of
  Statistics, 44 (2016), pp.~2497 -- 2532,
  \url{https://doi.org/10.1214/15-AOS1417}.

\end{thebibliography}

\end{document}